\documentclass[a4 paper, 11pt,twoside]{article}
\usepackage{fullpage}                           % Makes the page margins smaller to a predefined one.
\usepackage[lmargin=0.6in,rmargin=0.6in,tmargin=0.6in,headsep=.2in]{geometry}
\usepackage{graphicx}

\usepackage{times}
\usepackage{mathptmx}

\usepackage{forloop}
\usepackage{etoolbox}
\usepackage{calc}
\usepackage{wrapfig,lipsum,booktabs} % To produce the empty space on the left
\usepackage{xtab} % To produce the empty space on the left
\usepackage[dvipsnames]{xcolor}
\usepackage[normalem]{ulem}
\usepackage{sectsty}% http://ctan.org/pkg/sectsty
\allsectionsfont{\color{Brown}}
\makeatletter
\def\@seccntformat#1{\@ifundefined{#1@cntformat}%
{\csname the#1\endcsname\;}%  default
{\csname #1@cntformat\endcsname}% individual control
}
\def\section@cntformat{\thesection.\;} % Dot after the section number
\def\subsection@cntformat{\thesubsection.\;} % Dot after the subsection number
\makeatother
\usepackage{hyperref}  
\hypersetup{
    colorlinks=true,      
    urlcolor=blue,
    linkcolor=black,
    citecolor = black}
 
\usepackage{fancyhdr}
\usepackage{amsmath,amsfonts,amssymb,amsthm,epsfig,epstopdf,url,array}
 \usepackage[retainorgcmds]{IEEEtrantools}
 \usepackage{makeidx,epsfig,lscape}
 \usepackage{xcolor,pict2e}
\usepackage{upgreek}
\theoremstyle{definition}

\newtheorem{cor}{Corollary}
\newtheorem{thm}{Theorem}
\newtheorem{defin}{Definition}

\newtheorem{ce}{Consequence}

\newtheorem{theorem}{Theorem}

\newtheorem{lemma}[theorem]{Lemma}

\begin{document}
\vspace*{3cm}
%%%%%%%%%  TITLE %%%%%%%%%%%%%%%%%
{\noindent\huge\bf Martingales and Super-martingales Relative to a Convex Set of Equivalent Measures}\\[1cm]
%%%%%%%%%%%%%%%%  Author Data %%%%%%%%%%%%%%%%%%%
{\bf\large Nicholas S. Gonchar}\\[0.5cm]
Bogolyubov Institute for Theoretical Physics of NAS, Kiev, Ukraine\\
Email: mhonchar@i.ua\\
%%%%%%%%%%%   The Information Bar on the Left %%%%%%%%%%%

%%%%%%%%%%%%%%%  The Abstract and Keywords %%%%%%%%%%%%%%%%%%%%%%%%
{\noindent\color{Brown}\rule{0.7\textwidth}{2pt}}\\[0.2cm]
{\color{Brown}\bf\large Abstract}\\
In the paper, the martingales and super-martingales relative to a convex set of equivalent measures are systematically studied.
The notion of local regular super-martingale relative to a convex set of equivalent measures is introduced and the necessary and sufficient conditions of the local regularity of it in the discrete case are founded. 
The description of all local regular super-martingales relative to a convex set of equivalent measures is presented. The notion of the complete set of equivalent  measures is introduced. We prove that every bounded in some sense  super-martingale relative to the complete set of equivalent measures is local regular.
A new definition of the fair price of contingent claim in an incomplete market is given and the formula for the fair price of Standard Option of European type is found. The proved Theorems are the generalization of the famous Doob decomposition for super-martingale onto the case of super-martingales relative to a convex set of equivalent measures.
\vspace{0.5cm}\\
{\color{Brown}\bf\large Keywords}\\
 Random process; Convex set of equivalent measures; Optional Doob decomposition; Local regular super-martingale; martingale; Fair price of contingent claim.
\vspace{0cm}\\
{\color{Brown}\rule{0.7\textwidth}{2pt}}
%%%%%%%%%%%%%%%%%  The Document Starts Here %%%%%%%%%%%%%%
\section{Introduction}

In the paper, a new method of investigation of martingales and super-martingales 
relative to a convex set of equivalent measures is developed. A new proof  that the essential supremum over the set of regular martingales, generated by a certain 
nonnegative random value and a convex set of equivalent measures, is a super-martingale with respect to this set of measures, is given. 

A notion of local regular super-martingale is introduced and the necessary and sufficient conditions are found under that the above defined super-martingales are local regular ones. The last fact allowed us to describe the local regular super-martingales. It is proved that the existence of a nontrivial martingale relative to a convex set of equivalent measures, generally speaking, not guarantee for a nonnegative super-martingale to be a local regular one.

An important notion of the complete convex set  of equivalent measures is introduced. It is proved that any  super-martingale relative to  the complete convex set  of equivalent measures on a measurable space  with the finite set of elementary 
events is a local regular one. The notion of the complete convex set  of equivalent measures  is generalized onto an arbitrary space of elementary events. It is proved that the nonnegative and the majorized  from below super-martingales are local regular ones.

The definition of the fair price of contingent claim is introduced. The sufficient conditions of the existence of the fair price of contingent claim are presented. The conditions  that the introduced notion coincides with classical one are given. 

All these notions are used in the case as the convex set of equivalent measures is a set of equivalent martingale measures for the evolution of both risk and non risk assets.  The formula for the fair price of Standard Contract with Option of European type  in an incomplete market is found. 

 The notion of the complete convex set of equivalent measures permits us to give a new proof of the optional decomposition for a nonnegative super-martingale. This proof does not use the no-arbitrage arguments and  the measurable choice \cite{Kramkov}, \cite{FolmerKramkov1}, \cite{FolmerKabanov1},  \cite{FolmerKabanov}.

First, the optional decomposition for   diffusion processes super-martingale was opened by  by  El Karoui N. and  Quenez M. C. \cite{KarouiQuenez}. After that, Kramkov D. O. and Follmer H. \cite{Kramkov}, \cite{FolmerKramkov1} proved the optional decomposition for the nonnegative bounded super-martingales.  Folmer H. and Kabanov Yu. M.  \cite{FolmerKabanov1},  \cite{FolmerKabanov}  proved analogous result for an arbitrary super-martingale. Recently, Bouchard B. and Nutz M. \cite{Bouchard1} considered a class of discrete models and proved the necessary and sufficient conditions for the validity of the optional decomposition. 

The optional decomposition for super-martingales plays the fundamental role for the risk assessment in incomplete markets  \cite{Kramkov}, \cite{FolmerKramkov1},  \cite{KarouiQuenez},\cite{Gonchar2},  \cite{Gonchar555},  \cite{Gonchar557}, \cite{Honchar100}. 
Considered in the paper problem is a generalization of the corresponding one  that  appeared in mathematical finance  about the optional decomposition for a super-martingale and which is related with the construction of the superhedge strategy in incomplete financial markets.

Our statement of the problem unlike the above-mentioned one and it  is more general:  a super-martingale relative to a convex set of equivalent  measures is given  and it is necessary to find the conditions for the super-martingale and the set of measures under that  the optional decomposition exists.

The generality of our statement of the problem is that we do not require that the considered  set of measures was generated by the random process that is a local martingale  as it is done in the papers \cite{ Kramkov, FolmerKabanov, KarouiQuenez,  Bouchard1} and that is important for the proof of the  optional decomposition in these papers.
%%%%%%%%%%%%%%%%%%%%%%%%%%%%%%%%%%%%%%%%%%%%%
\section{Local regular super-martingales relative to  a convex set of equivalent measures.}

We assume that on a measurable space $\{\Omega,\mathcal{F}\}$ a filtration    ${\mathcal{F}_{m}\subset\mathcal{F}_{m+1}}\subset\mathcal{F}, \ m=\overline{0, \infty},$ and a family of convex set  of equivalent  measures $ M$ on $\mathcal{F}$ are given. Further, we assume that ${\cal F}_0=\{\emptyset, \Omega \}$ and the $\sigma$-algebra  ${\cal F}=\sigma(\bigvee\limits_{n=1}^\infty {\cal F}_n)$ is a minimal $\sigma$-algebra generated by the algebra $\bigvee\limits_{n=1}^\infty {\cal F}_n.$
A random process $\psi={\{\psi_{m}\}_{m=0}^{\infty}}$ is said to be adapted one relative to the filtration $\{{\cal F}_m\}_{m=0}^{\infty}$ if $\psi_{m}$ is a ${\cal F}_m$ measurable random value, $m=\overline{0,\infty}.$
\begin{defin}
An adapted random process  $f={\{f_{m}\}_{m=0}^{\infty}}$ is said to be   a super-martingale relative to the filtration ${\cal F}_m,\ m=\overline{0,\infty},$ and the  convex  family of equivalent  measures   $ M$  if $E^P|f_m|<\infty, \ m=\overline{1, \infty}, \ P \in M,$ and the inequalities 
\begin{eqnarray}\label{pk11} 
E^P\{f_m|{\cal F}_k\} \leq f_k, \quad 0 \leq k \leq m, \quad m=\overline{1, \infty}, \quad P \in M,
\end{eqnarray}
are valid.
\end{defin}
 Further,   for an adapted process  $f$ we use both the denotation $\{f_{m}, {\cal F}_m\}_{m=0}^{\infty} $ and the denotation  $\{f_{m}\}_{m=0}^{\infty}.$

\begin{defin} A  super-martingale $\{f_{m},\ {\cal F}_m\}_{m=0}^{\infty}$ relative to a convex set of equivalent measures M  is   a  local  regular one if  $ \sup\limits_{P \in M}E^P|f_m| < \infty, \ m=\overline{1, \infty},$  and  there exists  an adapted nonnegative increasing  random process  $\{g_{m},\ {\cal F}_m\}_{m=0}^{\infty}, \ g_0=0,$  \  $ \sup\limits_{P \in M}E^P|g_m| < \infty,\ m=\overline{1, \infty},$ such that  $\{f_m+g_m, \ {\cal F}_m\}_{m=0}^{\infty}$
is a martingale relative to every measure from $M.$
\end{defin}
The next elementary Theorem \ref{reww1} will be very useful later.

\begin{thm}\label{reww1} Let a super-martingale  $\{f_{m},\ {\cal F}_m\}_{m=0}^{\infty}, $ relative to a convex set of equivalent measures M  be such that   $ \sup\limits_{P \in M}E^P|f_m| < \infty, \ m=\overline{1, \infty}.$  The necessary and sufficient condition for  it   to be a local regular one is the existence of  an adapted nonnegative random process  $\{\bar g^0_{m},\ {\cal F}_m\}_{m=0}^{\infty},$  \  $ \sup\limits_{P \in M}E^P|\bar g^0_m| < \infty, \ m=\overline{1, \infty},$ such that
\begin{eqnarray}\label{o1}
f_{m-1} - E^P\{f_m|{\cal F}_{m-1}\}= E^P\{\bar g_m^0|{\cal F}_{m-1}\}, \quad m=\overline{1, \infty}, \quad P \in M.
\end{eqnarray}
\end{thm} 

\begin{proof}  Necessity. If   $\{f_{m},\ {\cal F}_m\}_{m=0}^{\infty}$ is  a local  regular super-martingale, then there exist a martingale  $\{\bar M_{m},\ {\cal F}_m\}_{m=0}^{\infty}$ and a non-decreasing nonnegative random process  $\{g_{m},\ {\cal F}_m\}_{m=0}^{\infty},$ $ \ g_0=0,$ such that
\begin{eqnarray}\label{reww3}
f_m = \bar M_m - g_m, \quad m=\overline{1, \infty}.
\end{eqnarray}
From here we  obtain  the equalities
$$E^P\{f_{m-1} -f_m|{\cal F}_{m-1}\}=$$
\begin{eqnarray}\label{reww4}
=E^P\{ g_m - g_{m-1}|{\cal F}_{m-1}\}=E^P\{\bar g_m^0|{\cal F}_{m-1}\}, \quad m=\overline{1, \infty} , \quad P \in M,
\end{eqnarray}
where we introduced the denotation $\bar g_m^0=g_m - g_{m-1} \geq 0.$
It is evident that $E^P\bar g_m^0\leq \sup\limits_{P\in M} E^Pg_m+\sup\limits_{P\in M} E^Pg_{m-1}< \infty.$

 Sufficiency. Suppose that there exists an adapted nonnegative random process $\bar g^0=\{\bar g_m^0\}_{m=0}^\infty, \ \bar g_0^0=0,$  $ E^P\bar g_m^0<\infty, \ m=\overline{1, \infty}, $ such that the equalities (\ref{o1}) hold.  Let us consider the  random process $\{\bar M_{m},\ {\cal F}_m\}_{m=0}^{\infty},$ where
\begin{eqnarray}\label{reww5}
\bar M_0=f_0, \quad \bar M_m=f_m+\sum\limits_{i=1}^m\bar g_m^0, \quad m=\overline{1, \infty}.
\end{eqnarray}
It is evident that $E^P|\bar M_m|< \infty$ and
\begin{eqnarray}\label{apm1}
E^P\{\bar M_{m-1} -  \bar M_m|{\cal F}_{m-1}\}=E^P\{f_{m-1} - f_m- \bar g_m^0|{\cal F}_{m-1}\}=0.
\end{eqnarray}
  Theorem  \ref{reww1} is proved.
\end{proof}

\begin{lemma}\label{l1} Any super-martingale  ${\{f_m, {\cal F}_m\}_{m=0}^{\infty}}$ relative to  a family of measures  $ M$ for which there hold equalities   $E^{P}f_{m}=f_{0}, \ m=\overline{1,\infty},$ \  ${ P\in M},$ is a martingale  with respect to this family of measures and the filtration   ${\cal F}_m,\ m=\overline{1,\infty}.$
\end{lemma}
\begin{proof} The proof of  Lemma \ref{l1} see \cite{Kallianpur}.\end{proof}
\section{Description of local regular super-martingales relative to a convex set of equivalent measures generated by the finite set of equivalent measures.}

Below, we describe the local regular super-martingales relative to a convex set of equivalent measures $M$ generated by the finite set of equivalent measures. For this we need some auxiliary statements.

\begin{lemma}\label{hj1}
On a measurable space $\{ \Omega, {\cal F}\}$ with filtration ${\cal F}_m$ on it,
let $G$ be a sub $\sigma$-algebra of  the $\sigma$-algebra ${\cal F}$ and let $f_s, s \in S,$ be a finite family of  nonnegative bounded  random values.  Then for every measure $P$ from $M$
\begin{eqnarray}\label{rgps1}
 E^P\{\max\limits_{s \in S}f_s|G\}\geq \max\limits_{s \in S}E^P\{f_s|G\}, \quad P \in M. \end{eqnarray}
\end{lemma}
\begin{proof} We have the inequalities
\begin{eqnarray}\label{gps1}
\max\limits_{s \in S} f_s \geq f_t, \quad t \in S.
\end{eqnarray}
Therefore,
\begin{eqnarray}\label{gps2}
E^P\{\max\limits_{s \in S} f_s|G\}\geq  E^P\{ f_t|G\}, \quad t \in S, \quad P \in M.
\end{eqnarray}
The last implies
\begin{eqnarray}\label{gps3}
E^P\{\max\limits_{s \in S} f_s|G\}\geq \max\limits_{s \in S} E^P\{f_s|G\}, \quad P \in M.
\end{eqnarray}
\end{proof}

In the next Lemma, we present the formula for calculation of the conditional expectation relative to another measure from $M.$
\begin{lemma}\label{q1}
 On a measurable space  $\{ \Omega, {\cal F}\}$ with a filtration ${\cal F}_n$ on it, 
let $M$ be a convex  set of equivalent measures  and let  $\xi$ be a bounded random value.  Then the following formulas 
\begin{eqnarray}\label{n1}
E^{P_1}\{\xi|{\cal F}_n\}=E^{P_2}\left\{\xi \varphi_n^{P_1}|{\cal F}_n\right\}, \quad n=\overline{1, \infty},   
\end{eqnarray}
are valid,  where

\begin{eqnarray}\label{apm2}
 \varphi_n^{P_1}=\frac{dP_1}{dP_2}\left[E^{P_2}\left\{\frac{dP_1}{dP_2}|{\cal F}_n\right\}\right]^{-1}, \quad P_1, \ P_2 \in M.
\end{eqnarray}
\end{lemma}
\begin{proof} The proof of  Lemma \ref{q1}   is evident.\end{proof}

Let $P_1, \ldots, P_k $ be a  family of  equivalent  measures on a measurable space   $\{ \Omega, {\cal F}\}$ and let us introduce  the denotation $M$ for a convex set of equivalent measures
\begin{eqnarray}\label{apm3}
 M=\left\{Q, \  Q=\sum\limits_{i=1}^{k}\alpha_i P_i, \ \alpha_i \geq 0, \  i=\overline{1,k},\ \sum\limits_{i=1}^{k}\alpha_i=1\right\}.
\end{eqnarray}
\begin{lemma}\label{q2}
If  $\xi$  is an  integrable random value relative to the set of equivalent  measures $P_1, \ldots, P_k $, then the formula
\begin{eqnarray}\label{n2}
\mathrm{ess}\sup\limits_{Q \in M}E^{Q}\{\xi|{\cal F}_n\}=\max\limits_{1 \leq i \leq k}E^{P_i}\{\xi|{\cal F}_n \}
\end{eqnarray}
is valid almost everywhere relative to the  measure $P_1$.
\end{lemma}
\begin{proof} The definition of $\mathrm{ess}\sup$ for non countable family of random variables see \cite{Chow}. Using the formula
\begin{eqnarray}\label{n3}
 E^{Q}\{\xi|{\cal F}_n\}=\frac{\sum\limits_{i=1}^{k}\alpha_i E^{P_1}\{\varphi_i|{\cal F}_n\}E^{P_i}\{\xi|{\cal F}_n\}}{\sum\limits_{i=1}^{k}\alpha_i E^{P_1}\{\varphi_i|{\cal F}_n\}}, \quad Q \in M,
\end{eqnarray}
 where $\varphi_i=\frac{dP_i}{dP_1},$ $Q=\sum\limits_{i=1}^{k}\alpha_i P_i,$ we obtain the inequality
\begin{eqnarray}\label{apm4} 
 E^{Q}\{\xi|{\cal F}_n\} \leq \max\limits_{1 \leq i \leq k}E^{P_i}\{\xi|{\cal F}_n\}, \quad Q \in M, 
\end{eqnarray}
or,
\begin{eqnarray}\label{apm5}
\mathrm{ess}\sup\limits_{Q \in M} E^{Q}\{\xi|{\cal F}_n\} \leq \max\limits_{1 \leq i \leq k}E^{P_i}\{\xi|{\cal F}_n\}.
\end{eqnarray}
On the other side  \cite{Chow}, 
\begin{eqnarray}\label{apm6} 
 E^{P_i}\{\xi|{\cal F}_n\} \leq \mathrm{ess}\sup\limits_{Q \in M} E^{Q}\{\xi|{\cal F}_n\}, \quad i=\overline{1,k}. 
\end{eqnarray}
Therefore,
\begin{eqnarray}\label{apm7} 
\max\limits_{1 \leq i \leq k} E^{P_i}\{\xi|{\cal F}_n\} \leq \mathrm{ess}\sup\limits_{Q \in M} E^{Q}\{\xi|{\cal F}_n\}. 
\end{eqnarray}
 Lemma \ref{q2} is proved.
\end{proof}

\begin{lemma}\label{q3} On  a measurable space $\{ \Omega, {\cal F}\}$ with a filtration ${\cal F}_n$ on it,  let $\xi$ be a nonnegative bounded random value. 
If $\frac{dP_i}{dP_l}, \  i, l=\overline{1,k},$ are ${\cal F}_1$ measurable  and $ P_1(\frac{dP_i}{dP_l}>0)=1,  \  i, l=\overline{1,k},$
 then the inequalities  
\begin{eqnarray}\label{0n3}
 E^{P_l}\{\max\limits_{1\leq i \leq k} E^{P_i}\{\xi|{\cal F}_n\}|{\cal F}_m\}\leq
 \max\limits_{1\leq i \leq k}E^{P_i}\{\xi |{\cal F}_m\}, \ l=\overline{1,k},  \  n>m,   
\end{eqnarray}
are valid.
\end{lemma}
\begin{proof}  From   Lemma \ref{q1} and Lemma \ref{q3} conditions relative to the density of one measure with respect to another, we have
\begin{eqnarray}\label{apm9}
\max\limits_{1\leq i \leq k}E^{P_i}\{\xi|{\cal F}_n\}= E^{P_l}\{\xi |{\cal F}_n\}, \quad l=\overline{1, k}.   
\end{eqnarray}
From the equality (\ref{apm9}) we obtain the inequality
\begin{eqnarray}\label{apm11}
E^{P_l}\{\max\limits_{1\leq i \leq k}E^{P_i}\{\xi|{\cal F}_n\}|{\cal F}_m\}\leq
\max\limits_{1\leq i \leq k}E^{P_i}\{\xi |{\cal F}_m\}, \quad l=\overline{1, k}. 
\end{eqnarray}
Lemma \ref{q3} is proved.
\end{proof}

In this section, we assume that the conditions of Lemma  \ref{q3} relative 
to the density of one measure with respect to another are true.

\begin{lemma}\label{q5} On a measurable space $\{ \Omega, {\cal F}\}$ with a filtration ${\cal F}_n$ on it,
let  $\xi$ be  a nonnegative  random value which is  integrable  relative to the set of  equivalent measures $P_1, \ldots, P_k$ . 
Then the inequalities
\begin{eqnarray}\label{start6}
 E^Q\{\max\limits_{1 \leq i \leq k} E^{P_i}\{\xi|{\cal F}_n\}|{\cal F}_m\} \leq 
\max\limits_{1 \leq i \leq k} E^{P_i}\{\xi|{\cal F}_m\}, \quad n>m, \quad Q \in M,
\end{eqnarray}
are valid.
\end{lemma}
\begin{proof} Using Lemma \ref{q3} inequalities  for the nonnegative bounded
$\xi$    and the formula
\begin{eqnarray}\label{port5}
 E^{Q}\{\Phi|{\cal F}_{m}\}=\frac{\sum\limits_{i=1}^{k}\alpha_i E^{P_1}\{\varphi_i|{\cal F}_{m}\}E^{P_i}\{\Phi|{\cal F}_{m}\}}{\sum\limits_{i=1}^{k}\alpha_i E^{P_1}\{\varphi_i|{\cal F}_{m}\}}, \quad Q \in M,
\end{eqnarray}
  where $\Phi=\max\limits_{1\leq i \leq k}E^{ P_i}\{\xi|{\cal F}_n\},  \varphi_i=\frac{dP_i}{dP_1}, \ i=\overline{1,k},$
 we prove  Lemma \ref{q5} inequalities.

Let us consider the case, as 
$\max\limits_{1 \leq i \leq k}  E^{P_i}\xi < \infty.$ Let $\xi_s, s=\overline{1, \infty},$ be a sequence of bounded random values  converging to $\xi$ monotonuosly. Then 
\begin{eqnarray}\label{alkn44}
 E^{Q}\{\max\limits_{1 \leq i \leq k} E^{P_i}\{\xi_s|{\cal F}_n\}|{\cal F}_m\}\leq
\max\limits_{1 \leq i \leq k}  E^{Q}\{\xi_s |{\cal F}_m\}, \quad l=\overline{1,k}.
\end{eqnarray}
Due to  the monotony convergence  of $\xi_s$ to $\xi,$ as $s \to \infty,$ we can pass to the limit under the conditional expectations on the left and  right sides in the inequalities (\ref{alkn44}) that proves   Lemma \ref{q5}.  
\end{proof}
 
\begin{lemma}\label{lkq4} 
 On a measurable space $\{ \Omega, {\cal F}\}$ with filtration ${\cal F}_n$ on it,  for every nonnegative integrable random value $\xi$ relative to a set of equivalent measures $\{P_1, \ldots,P_k\}$ the inequalities 
\begin{eqnarray}\label{lkn4}
 E^{Q}\{\mathrm{ess}\sup\limits_{P\in M} E^{P}\{\xi |{\cal F}_n\}|{\cal F}_m\}\leq
\mathrm{ess}\sup\limits_{P\in M} E^{P}\{\xi |{\cal F}_m\}, \quad Q \in M, \quad n>m, 
\end{eqnarray}
are  valid.
\end{lemma}

Lemma  \ref{lkq4} is a consequence of Lemma \ref{q5}.

\begin{lemma}\label{1q5} 
On a measurable space $\{ \Omega, {\cal F}\}$ with a filtration ${\cal F}_m$ on it, 
let  $\xi$ be a nonnegative integrable random value with respect to a set of equivalent measures $\{P_1, \ldots, P_k \}$  and such that 
\begin{eqnarray}\label{r7}
 E^{P_i}\xi=M_0, \quad i=\overline{1,k},
\end{eqnarray}
then the random process  $\{M_m=\mathrm{ess}\sup\limits_{P\in M}E^P\{\xi|{\cal F}_m\}, {\cal F}_m\}_{m=0}^\infty$ is a martingale relative to a convex set of equivalent measures  $M.$
\end{lemma}
\begin{proof} Due to Lemma  \ref{lkq4}, a random process\\ $\{M_m=\mathrm{ess}\sup\limits_{P\in M}E^P\{\xi|{\cal F}_m\}, {\cal F}_m\}_{m=0}^\infty $  is a super-martingale, that is, 
\begin{eqnarray}\label{apm19} 
E^P\{M_m | {\cal F}_{m-1}\} \leq M_{m-1}, \quad m=\overline{1, \infty}, \quad P \in M.
\end{eqnarray}
Or, $E^PM_m \leq M_0.$
From the other side,

\begin{eqnarray}\label{apm20} 
 E^{P_s} [ \max\limits_{1\leq i \leq k}E^{P_i}\{\xi|{\cal F}_m\}] \ge 
\max\limits_{1\leq i \leq k}E^{P_s}E^{P_i}\{\xi|{\cal F}_m\} \geq M_0, \quad s=\overline{1,k}.
\end{eqnarray}
 The above inequalities imply $E^{P_s}M_m= M_0, \  m=\overline{1, \infty}, \ s=\overline{1, k}.$
The last equalities lead to the  equalities $E^{P}M_m= M_0, \ m=\overline{1, \infty}, \ P \in M.$
The fact that $M_m$ is a super-martingale relative to the set of measures $M$ and the above equalities  prove  Lemma \ref{1q5}, since the Lemma \ref{l1} conditions are valid.
\end{proof}
In the next Theorem we denote ${\cal F}=\sigma(\bigvee\limits_{i=1}^\infty{\cal F}_i)$
the minimal $\sigma$-algebra generated by the algebra $\bigvee\limits_{i=1}^\infty{\cal F}_i.$
\begin{thm}\label{mars12}
 Let  $\{ \Omega, {\cal F}\}$ be  a measurable space  with a filtration ${\cal F}_m$ on it and let $\xi$ be a nonnegative integrable random value with respect to a set of equivalent  measures $P_1, \ldots, P_k.$ 
The necessary and sufficient conditions of the local regularity of
 the super-martingale $\{f_m, {\cal F}_m\}_{m=0}^\infty, $ where 
\begin{eqnarray}\label{marsss13}
f_m=\mathrm{ess}\sup\limits_{P \in M}E^P\{\xi| {\cal F}_m\}, \quad  m=\overline{1,\infty}, \quad   \max\limits_{1 \leq i \leq k}E^{P_i}\xi< \infty, 
\end{eqnarray}
is  its uniform integrability relative to the set of measure $P_1, \ldots, P_k$ and the fulfillment  of the equalities
\begin{eqnarray}\label{mars13}
E^{P_i}\xi=f_0, \quad i=\overline{1,k}.
\end{eqnarray}
\end{thm}
\begin{proof} The necessity. Let $\{f_m, {\cal F}_m\}_{m=0}^\infty $  be a local regular super-martingale. Then 
\begin{eqnarray}\label{apm21} 
f_n=M_n-g_n, \quad n=\overline{0,\infty}, \quad g_0=0,\quad  f_0=E^{P_i}M_n,\quad i=\overline{1,k}.
\end{eqnarray}
From here we obtain $E^{P_i}g_n\leq f_0, \ i=\overline{1,k}.$ Due to the  uniform integrability of $f_n$ and $g_n$ we obtain
\begin{eqnarray}\label{NS1} 
E^{P_i}(f_\infty+g_\infty)=f_0, \quad i=\overline{1,k},
\end{eqnarray}
where $f_\infty=\xi,$ $g_\infty=\lim\limits_{n\to \infty}g_n,$    since ${\cal F}=\sigma(\bigvee\limits_{i=1}^\infty{\cal F}_i).  $ But $f_0=\max\limits_{1\leq i\leq k}E^{P_i}\xi=E^{P_{i_0}}\xi.$ From (\ref{NS1}) we have $E^{P_{i_0}}g_\infty=0.$
 The last  equality  gives $g_\infty=0,$ or
\begin{eqnarray}\label{apm22} 
 E^{P_i}\xi=E^{P_{i_0}}\xi, \quad i=\overline{1,k}.
\end{eqnarray}
 The sufficiency. If the conditions  of  Theorem \ref{mars12} are satisfied,  then  $\{\bar M_m, {\cal F}_m\}_{m=0}^\infty $ is a martingale, where  $\bar M_m= \sup\limits_{P \in M}E^P\{\xi|{\cal F}_m\}.$ The last implies the local regularity of $\{f_m, {\cal F}_m\}_{m=0}^\infty.$
 Theorem \ref{mars12} is proved.
\end{proof}

\section{Description of local regular super-martingales relative to an  arbitrary convex  set of equivalent measures.} 
Below, in the paper we assume  that  an arbitrary convex set of equivalent measures $M$ on a measurable space $\{\Omega, {\cal F}\}$  and a filtration ${\cal F}_n$ on it satisfies the conditions: the density  $ \frac{dP}{dQ}$ is ${\cal F}_1$ measurable  one and $P_0(\frac{dP}{dQ}>0)=1$ for all $P , Q \in M,$ where the fixed measure $P_0 \in
M.$  Such a class  of equivalent measures is sufficiently wide. It contains the class of equivalent martingale measures generated by a local martingale.

 Introduce into consideration a set  $A_0$ of all integrable  nonnegative random values $\xi$  relative to a convex set of equivalent measures $M$ satisfying conditions
\begin{eqnarray}\label{0mars6}
E^P\xi=1, \quad P \in M.
\end{eqnarray}
It is evident that the set  $A_0$ is not empty, since contains  the random value $\xi =1.$
More interesting case is  as  $A_0$ contains more then one element.

\begin{lemma}\label{tmars5}  On a  measurable space $\{\Omega, {\cal F}\}$ and a filtration ${\cal F}_n$ on it, let $M$ be an arbitrary convex set of equivalent 
measures. If  the  nonnegative random value $\xi$ is such that $\sup\limits_{P \in M}E^P\xi< \infty,$  then 
$\{f_m=\mathrm{ess}\sup\limits_{P \in M}E^P\{\xi| {\cal F}_m\},  {\cal F}_m\}_{m=0}^\infty $ is a super-martingale relative to the convex set of equivalent measures $M.$
\end{lemma}
\begin{proof} From the definition of $\mathrm{ess}\sup$ \cite{Chow}, for every $\mathrm{ess}\sup\limits_{P \in M}E^P\{\xi| {\cal F}_m\}$ there exists a countable set $D_m$ such that
\begin{eqnarray}\label{apm23}
\mathrm{ess}\sup\limits_{P \in M}E^P\{\xi| {\cal F}_m\}=\sup\limits_{P \in D_m}E^P\{\xi| {\cal F}_m\}, \quad m=\overline{0, \infty}.
\end{eqnarray} 
 The set $D=\bigcup\limits_{m=0}^\infty D_m$ is also countable  one and the equality
\begin{eqnarray}\label{tamar1}  
\mathrm{ess}\sup\limits_{P \in M}E^P\{\xi| {\cal F}_m\}=\sup\limits_{P \in D}E^P\{\xi| {\cal F}_m\}
\end{eqnarray}
is true. 
Really, since
\begin{eqnarray}\label{tamar2}
 \sup\limits_{P \in D}E^P\{\xi| {\cal F}_m\} \geq \sup\limits_{P \in D_m}E^P\{\xi| {\cal F}_m\}=\mathrm{ess}\sup\limits_{P \in M}E^P\{\xi| {\cal F}_m\}.
\end{eqnarray}
From the other side,
\begin{eqnarray}\label{tamar3}
\mathrm{ess}\sup\limits_{P \in M}E^P\{\xi| {\cal F}_m\} \geq E^Q\{\xi| {\cal F}_m\}, \quad Q \in M.
\end{eqnarray}
The last gives 
\begin{eqnarray}\label{tamar4}
\mathrm{ess}\sup\limits_{P \in M}E^P\{\xi| {\cal F}_m\} \geq \sup\limits_{P \in D}E^P\{\xi| {\cal F}_m\}. 
\end{eqnarray}
The inequalities (\ref{tamar2}),  (\ref{tamar4}) prove the needed statement. So, for all $m$ we can choose the common set $D.$ Let $D=\{\bar P_1,\ldots \bar P_n, \ldots\}.$ Due to Lemma 
\ref{lkq4}, for every $Q \in \bar M_k,$
we have
\begin{eqnarray}\label{tamar6}
 E^Q\{\max\limits_{1 \leq i \leq k} E^{\bar P_i}\{\xi|{\cal F}_n\}|{\cal F}_m\} \leq 
\max\limits_{1 \leq i \leq k} E^{\bar P_i}\{\xi|{\cal F}_m\}, \quad n>m, \quad Q \in \bar M_k,
\end{eqnarray}
 where
\begin{eqnarray}\label{tamar5}
\bar M_k=\{P \in M, P=\sum\limits_{i=1}^k\alpha_i \bar P_i, \  \alpha_i \geq 0, \ \sum\limits_{i=1}^k\alpha_i=1\}.
\end{eqnarray}
It is evident that $\max\limits_{1 \leq i \leq k} E^{\bar P_i}\{\xi|{\cal F}_n\}$ tends to $ \sup\limits_{P \in D}E^P\{\xi| {\cal F}_n\}$ monotonously increasing, as $k \to \infty.$ Fixing $Q \in \bar M_k \subset \bar M_{k+1}$ and tending $k $ to the infinity in 
the inequalities (\ref{tamar6}), we obtain
\begin{eqnarray}\label{tamar7}
 E^Q\{\sup\limits_{P \in D}E^P\{\xi| {\cal F}_n\}| {\cal F}_m\} \leq 
\sup\limits_{P \in D}E^P\{\xi| {\cal F}_m\}, \quad n>m, \quad Q \in \bar M_k.
\end{eqnarray}
The last inequalities implies that for every measure $Q,$ belonging to the convex span, constructed on the set $D,$ $\{f_m=\mathrm{ess}\sup\limits_{P \in M}E^P\{\xi| {\cal F}_m\},  {\cal F}_m\}_{m=0}^\infty $ is a super-martingale relative to the convex set of equivalent measures, generated by the set  $D.$ Now, if a measure $Q_0$ does not  belong to the convex span, constructed on the set $D,$ then we can  add it to the set $D$ and repeat the proof made above. As a result, we proved that  $\{f_m=\mathrm{ess}\sup\limits_{P \in M}E^P\{\xi| {\cal F}_m\},  {\cal F}_m\}_{m=0}^\infty $ is also a super-martingale relative to the measure $Q_0.$  Zorn Lemma \cite{Kelley} complete the proof of  Lemma \ref{tmars5}.
\end{proof}

\begin{thm}\label{fmars5} On a measurable space $\{\Omega, {\cal F}\}$ and a filtration ${\cal F}_n$ on it, let $M$ be an arbitrary convex set of equivalent 
measures.  For a random value $\xi \in A_0,$  the random process $\{ E^P\{\xi|{\cal F}_m\}, {\cal F}_m\}_{m=0}^\infty,$ $P \in M,$  is a local regular martingale  relative to the convex set of equivalent 
measures $M.$ 
\end{thm}
\begin{proof}  Let $P_1, \ldots, P_n$ 
be a certain subset of measures from $M.$ Denote 
$M_n$  a convex set of equivalent measures 
\begin{eqnarray}\label{mars8}
M_n = \{P \in M, \  P=\sum\limits_{i=1}^n\alpha_i P_i, \  \alpha_i \geq 0, \ i=\overline{1,n}, \ \sum\limits_{i=1}^n\alpha_i=1\}.
\end{eqnarray}
Due to Lemma \ref{1q5}, $\{\bar M_m, {\cal F}_m\}_{m=0}^\infty$ is a  martingale relative
to the set of measures $M_n,$ where $ \bar M_m=\mathrm{ess}\sup\limits_{P \in M_n}E^P\{\xi|{\cal F}_m\}, \ \xi \in A_0.$  Let us consider an arbitrary measure $P_0 \in M$ and let
\begin{eqnarray}\label{mars9}
M_n^{P_0} = \{P \in M, \ P= \sum\limits_{i=0}^n\alpha_i P_i, \  \alpha_i \geq 0, \  i=\overline{0,n}, \ \sum\limits_{i=0}^n\alpha_i=1\}.
\end{eqnarray}
Then  $\{\bar M_m^{P_0}, {\cal F}_m\}_{m=0}^\infty,$ where $\bar M_m^{P_0}=\mathrm{ess}\sup\limits_{P \in M_n^{P_0}}E^P\{\xi|{\cal F}_m\},$ is a martingale relative to the set of measures $M_n^{P_0}.$ It is evident that
\begin{eqnarray}\label{mars10}
\bar M_m \leq \bar M_m^{P_0}, \quad   m=\overline{0, \infty}.
\end{eqnarray}
Since $E^P\bar M_m=E^P\bar M_m^{P_0}=1, \ m=\overline{0, \infty}, \ P \in M_n,$ the inequalities (\ref{mars10}) give $\bar M_m=\bar M_m^{P_0}.$ Analogously,
$E^{P_0}\{\xi|{\cal F}_m\} \leq \bar M_m^{P_0}.$ From the equalities 
$ E^{P_0}E^{P_0}\{\xi|{\cal F}_m\}$ $ = E^{P_0}\bar M_m^{P_0}=1$ we obtain
$E^{P_0}\{\xi|{\cal F}_m\} = \bar M_m^{P_0}=\bar M_m.$ 
Since the measure $P_0$ is an  arbitrary one it implies that $\{E^P\{\xi|{\cal F}_m\}, {\cal F}_m\}_{m=0}^\infty$ is a martingale relative to all measures from $M.$
Due to Theorem  \ref{reww1}, it is a local regular super-martingale with the random process $\bar g^0_m=0,  m=\overline{0, \infty}. $
 Theorem \ref{fmars5} is proved.
\end{proof}

\begin{thm}\label{mmars1} On a measurable space $\{\Omega, {\cal F}\}$ and a filtration ${\cal F}_n$ on it, let $M$ be an arbitrary convex set of equivalent 
measures. 
If $\{f_m, {\cal F}_m\}_{m=0}^\infty$ is an adapted random process  satisfying conditions
\begin{eqnarray}\label{mmars2}
f_m \leq f_{m-1}, \quad  E^P\xi|f_m| <\infty, \quad P \in M \quad m=\overline{1, \infty}, \quad  \xi \in A_0,
\end{eqnarray}
then the random process
\begin{eqnarray}\label{mmars3}
 \{ f_mE^P\{\xi|{\cal F}_m\}, {\cal F}_m\}_{m=0}^\infty, \quad  P \in M,
\end{eqnarray}
is a local regular super-martingale relative to the convex set of equivalent  measures  $M.$
\end{thm}
\begin{proof} Due to Theorem  \ref{fmars5}, the random process 
$\{ E^P\{\xi|{\cal F}_m\}, {\cal F}_m\}_{m=0}^\infty$ is a martingale relative to the convex set of equivalent  measures  $M.$ Therefore,
\begin{eqnarray*}
f_{m-1}E^P\{\xi|{\cal F}_{m-1}\} - E^P\{ f_m E^P\{\xi|{\cal F}_m\}|{\cal F}_{m-1}\}=
\end{eqnarray*}
\begin{eqnarray}\label{mmars4}
E^P\{ (f_{m-1} - f_m) E^P\{\xi|{\cal F}_m\}|{\cal F}_{m-1}\}, \quad m=\overline{1, \infty}.
\end{eqnarray}
So, if to put  $\bar g_m^0= (f_{m-1} - f_m) E^P\{\xi|{\cal F}_m\}, \ m=\overline{1, \infty}, $ then $\bar g_m^0 \geq 0,$  it is  ${\cal F}_m$-measurable and
$E^P\bar g_m^0 \leq E^P\xi(|f_{m-1}|+|f_m|)< \infty.$ It proves the needed statement.
\end{proof}

\begin{cor}\label{hg1} If $f_m=\alpha, \ m=\overline{1, \infty}, \ \alpha \in R^1,$ $\xi \in A_0,$ then 
$\{\alpha E^P\{\xi |{\cal F}_m\},  {\cal F}_m \}_{m=0}^\infty$ is a local regular martingale. Assume that  $\xi =1,$ then $\{f_m, {\cal F}_m\}_{m=0}^\infty$ is a local regular super-martingale relative to a convex set of equivalent  measures  $M.$ 
\end{cor}

Denote  $F_0$ the set of adapted processes
\begin{eqnarray}\label{mmars5}
F_0=\{ f=\{f_m\}_{m=0}^\infty,  \   P(|f_m| <\infty) =1, \ P \in M, \ f_m \leq f_{m-1}\}.
\end{eqnarray}
For every $\xi \in A_0$ let us introduce the set of adapted processes
$$ L_{\xi}=$$
\begin{eqnarray}\label{mmars6}
\{\bar f=\{f_mE^P\{\xi|{\cal F}_m\}\}_{m=0}^\infty, \  \{f_m\}_{m=0}^\infty \in F_0, \   E^P\xi|f_m| <\infty, \ P \in M\},
\end{eqnarray}
and 
\begin{eqnarray}\label{mmars7}
V=\bigcup\limits_{\xi \in A_0}L_{\xi}.
\end{eqnarray}

\begin{cor}\label{fdr1} Every  random  process from the set $K,$ where
\begin{eqnarray}\label{mmars88}
K=\left \{ \sum\limits_{i=1}^mC_i \bar f_i, \ \bar f_i \in V, \  C_i \geq 0, \ i=\overline{1, m}, \ m=\overline{1, \infty}\right\}, 
\end{eqnarray}
 is a local regular super-martingale relative to the convex set of equivalent  measures  $M$  on a measurable space $\{\Omega, {\cal F}\}$ with filtration ${\cal F}_m$ on it. 
\end{cor}
\begin{proof} The proof is evident.
\end{proof}

\begin{thm}\label{mmars9} On a measurable space $\{\Omega, {\cal F}\}$ and a filtration ${\cal F}_n$ on it, let $M$ be an arbitrary convex set of equivalent 
measures. 
 Suppose that  $\{f_m, {\cal F}_m\}_{m=0}^\infty$ is a nonnegative uniformly integrable super-martingale relative to a convex set of equivalent  measures  $M,$ then 
the necessary and sufficient conditions for it  to be a local regular one is belonging it to the set $K.$
\end{thm}
\begin{proof}
Necessity.  It is evident  that if  $\{f_m, {\cal F}_m\}_{m=0}^\infty$ belongs to $K,$ then it is a local regular super-martingale.

 Sufficiency. Suppose that  $\{f_m,{\cal F}_m\}_{m=0}^\infty$ is a local regular super-martingale. Then there exists  nonnegative adapted process $\{\bar g_m^0\}_{m=1}^ \infty,  \ E^P\bar g_m^0< \infty, \ m=\overline{1, \infty}, $  and a martingale  $\{M_m\}_{m=0}^ \infty,$
such that 
\begin{eqnarray}\label{mmars8}
f_m=M_m - \sum\limits_{i=1}^m\bar g_i^0, \quad  m=\overline{0, \infty}. 
\end{eqnarray}
Then $M_m \geq 0, \ m=\overline{0, \infty}, \ E^P M_m <\infty, \ P\in M.$
Since $0< E^PM_m=f_0< \infty$ we have $E^P\sum\limits_{i=1}^m\bar g_i^0< f_0.$ Let us put $g_{\infty}=\lim\limits_{m \to \infty}\sum\limits_{i=1}^m\bar g_i^0.$
Using the uniform integrability of $f_m,$ we can pass to the limit in the equality
\begin{eqnarray}\label{apm23}
 E^P(f_m +\sum\limits_{i=1}^m\bar g_i^0)=f_0, \quad P \in M,
\end{eqnarray}
as $m \to \infty$.
Passing to the limit in the last equality, as $m \to \infty,$  we obtain
\begin{eqnarray}\label{apm23}
E^P(f_\infty +g_{\infty})=f_0, \quad P \in M.
\end{eqnarray}
Introduce into consideration a  random value $\xi=\frac{f_\infty +g_{\infty}}{f_0}.$
Then $E^P\xi=1, \ P \in M.$   From here we obtain that  $\xi \in A_0$ and 
\begin{eqnarray}\label{apm24}
M_m=f_0E^P\{\xi|{\cal F}_m\}, \ m=\overline{0, \infty}.
\end{eqnarray}
 Let us put $\bar f_m^2=-\sum\limits_{i=1}^m\bar g_i^0. $ It is easy to see that the adapted random process $\bar f_2=\{\bar f_m^2, {\cal F}_m\}_{m=0}^\infty$ belongs to $F_0.$ Therefore,
for the super-martingale $f=\{ f_m,{\cal F}_m\}_{m=0}^\infty$ the representation 
$$f=\bar f_1+ \bar f_2,$$
is valid, where $\bar f_1=\{f_0E^P\{\xi|{\cal F}_m\}, {\cal F}_m\}_{m=0}^ \infty$  belongs to $L_{\xi}$
with  $ \xi = \frac{f_\infty +g_{\infty}}{f_0}$  and $ f_m^1=f_0, \ m=\overline{0,\infty}.$ The same is valid for $\bar f_2$ with $\xi=1.$ This implies that $f$ belongs to the set $K.$  Theorem 
\ref{mmars9} is proved.
\end{proof}
\begin{thm}\label{9mmars9} On a measurable space $\{\Omega, {\cal F}\}$ and a filtration ${\cal F}_n$ on it, let $M$ be an arbitrary convex set of equivalent 
measures. Suppose that  the super-martingale $\{f_m, {\cal F}_m\}_{m=0}^\infty$   relative to the convex set of equivalent  measures  $M$ satisfy conditions
\begin{eqnarray}\label{99mmars88}
|f_m|\leq C \xi_0, \quad m=\overline{1, \infty}, \quad \xi_0 \in A_0, \quad 0<C<\infty,
\end{eqnarray}
 then the necessary and sufficient conditions for it  to be a local regular one is belonging it to the set $K.$
\end{thm}
\begin{proof} The necessity is evident.\\
 Sufficiency.  
 Suppose that  $\{f_m,{\cal F}_m\}_{m=0}^\infty$ is a local regular super-martingale. Then there exists a nonnegative adapted random process $\{\bar g_m^0\}_{m=1}^ \infty,  \ E^P\bar g_m^0< \infty, \ m=\overline{1, \infty}, $  and a martingale  $\{M_m\}_{m=0}^ \infty, \ E^P| M_m| <\infty, \  m=\overline{1, \infty}, \  P\in M,$
such that 
\begin{eqnarray}\label{9mmars8}
f_m=M_m - \sum\limits_{i=1}^m\bar g_i^0, \quad  m=\overline{0, \infty}. 
\end{eqnarray}
The inequalities $f_m +C\xi_0 \geq 0, \  m=\overline{1, \infty}, $ give the inequalities
\begin{eqnarray}\label{8mmars8}
 f_m+C E^P\{\xi_0|{\cal F}_m\}\geq 0, \quad  m=\overline{0, \infty}.
\end{eqnarray}
From the inequalities (\ref{99mmars88}) it follows that the super-martingale $\{f_m,{\cal F}_m\}_{m=0}^\infty$ is a uniformly integrable one  relative  to the convex set of  equivalent measures $M$.  The martingale  $\{E^P\{\xi_0|{\cal F}_m\}, {\cal F}_m\}_{m=0}^\infty$ relative to the convex set of  equivalent measures $M$ is also uniformly integrable one.

Then $M_m+C E^P\{\xi_0|{\cal F}_m\} \geq 0, \ m=\overline{0, \infty}. $
Since $0< E^P[M_m+C E^P\{\xi_0|{\cal F}_m\}]=f_0+C< \infty$ we have $E^P\sum\limits_{i=1}^m\bar g_i^0< f_0+C.$ Let us put $g_{\infty}=\lim\limits_{m \to \infty}\sum\limits_{i=1}^m\bar g_i^0.$
Using the uniform integrability of $f_m$ and $\sum\limits_{i=1}^m\bar g_i^0$ we can pass to the limit in the equality
\begin{eqnarray}\label{apm28}
 E^P(f_m +\sum\limits_{i=1}^m\bar g_i^0)=f_0, \quad P \in M,
\end{eqnarray}
as $m \to \infty$.
Passing to the limit in the last equality, as $m \to \infty,$  we obtain
\begin{eqnarray}\label{apm24}
E^P(f_\infty +g_{\infty})=f_0, \quad P \in M.
\end{eqnarray}
Introduce into consideration a  random value $\xi_1=\frac{f_\infty+C \xi_0+g_{\infty}}{f_0+C}\geq 0.$
Then $E^P\xi_1=1, \ P \in M.$ From here we obtain that  $\xi_1 \in A_0$ and for the super-martingale  $f=\{ f_m,{\cal F}_m\}_{m=0}^\infty$   the representation 
\begin{eqnarray}\label{apm25}
 f_m=f_m^0E^P\{\xi_0|{\cal F}_m\}+f_m^1 E^P\{\xi_1|{\cal F}_m\}+f_m^2E^P\{\xi_2|{\cal F}_m\} , \ m=\overline{0, \infty},
\end{eqnarray}
is valid, where $f_m^0=-C, \ f_m^1=f_0+C, \ f_m^2=-\sum\limits_{i=1}^m\bar g_i^0, \ m=\overline{0, \infty}, \ \xi_2=1.$
From  the last representation it follows that  the super-martingale $f=\{ f_m,{\cal F}_m\}_{m=0}^\infty$ belongs to the set $K.$ Theorem 
\ref{9mmars9} is proved.
\end{proof}

\begin{cor}\label{mars16}  Let $f_N, \  N< \infty,$ be a ${\cal F}_N$-measurable integrable random value,  $\sup\limits_{P \in M} E^P|f_N| < \infty,$ and let there exist $\alpha_0 \in R^1$ such that
$$ -\alpha_0  M_N+ f_N \leq 0, \quad \omega \in \Omega, $$
where $\{ M_m, {\cal F}_m\}_{m=0}^\infty=\{E^P\{\xi|{\cal F}_m\}, {\cal F}_m\}_{m=0}^\infty, \ \xi \in A_0. $ 
Then a  super-martingale $\{ f_m^0+ \bar f_m\}_{m=0}^\infty$ is a local regular one relative to  the convex set of equivalent measures $M,$ where
\begin{eqnarray}\label{apm26}
f_m^0=\alpha_0 M_m,  
\end{eqnarray}
\begin{eqnarray}\label{apm27}
\bar f_m=
\left\{\begin{array}{l l} 0, & m<N, \\
f_N - \alpha_0 M_N, & m \geq N.  
 \end{array} \right. 
\end{eqnarray}
\end{cor}
\begin{proof}  It is evident that $\bar f_{m-1} -\bar f_m \geq 0, \  m=\overline{0,\infty}.$
Therefore, the super-martingale
\begin{eqnarray}\label{apm27}
f_m^0+ \bar f_m=
\left\{\begin{array}{l l} \alpha_0 M_m, & m<N, \\
f_N , & m= N,  \\
f_N - \alpha_0 M_N+\alpha_0 M_m, & m>N
 \end{array} \right. 
\end{eqnarray}
is a local regular one relative to  the convex set of equivalent measures $M.$
 Corollary \ref{mars16} is proved.
\end{proof}

\section{Optional decomposition for  super-martingales relative to the complete convex set of equivalent measures.}

 In this section we introduce the notion of complete set of equivalent measures and prove that  non negative super-martingales are local regular ones with respect to this set of measures. For this purpose  we are needed the next auxiliary statement.
\begin{thm}\label{nick1}
The necessary and sufficient conditions of the local regularity of the nonnegative super-martingale $\{f_m, {\cal F}_m\}_{m=0}^\infty$ relative to a convex set of equivalent measures $M$ are the existence of ${\cal F}_m$-measurable random values $\xi_m^0 \in A_0, \  m=\overline{1, \infty},$ such that
\begin{eqnarray}\label{nick2}
\frac{f_m}{f_{m-1}} \leq \xi_m^0, \quad E^P\{\xi_m^0|{\cal F}_{m-1}\}=1, \quad P\in M, \quad m=\overline{1, \infty}.
\end{eqnarray}
\end{thm}
\begin{proof} The necessity.  Without loss of generality, we assume that $f_m\geq a$ for a certain real number $a>0.$ Really, if it is not so, then we can come to the consideration of the super-martingale  $\{f_m+a, {\cal F}_m\}_{m=0}^\infty.$  Thus,   let  $\{f_m, {\cal F}_m\}_{m=0}^\infty$ be a   nonnegative  local regular super-martingale. Then there exists a nonnegative adapted random process
$\{g_m\}_{m=0}^\infty, \ g_0=0,$ such that $\sup\limits_{P \in M}E^Pg_m<\infty,$
\begin{eqnarray}\label{nick3}
f_{m-1} - E^P\{f_m|{\cal F}_{m-1}\} =E^P\{g_m|{\cal F}_{m-1}\}, \quad P \in M, \quad m=\overline{1, \infty}.
\end{eqnarray}
Let us put $\xi_m^0=\frac{f_m+g_m}{f_{m-1}}, \ m=\overline{1, \infty}.$ Then  $\xi_m^0 \in A_0 $ and 
from the equalities  (\ref{nick3}) we obtain $E^P\{\xi_m^0|{\cal F}_{m-1}\}=1, \ P \in M, \ m=\overline{1, \infty}.$
It is evident that the inequalities (\ref{nick2}) are valid.

 The sufficiency. Suppose that the conditions of  Theorem \ref{nick1} are valid.
Then   $ f_m \leq f_{m-1}+ f_{m-1}(\xi_m^0 -1).$
Introduce the  denotation  $g_m= -f_m+ f_{m-1}\xi_m^0.$ Then  $g_m \geq 0, $  $\sup\limits_{P \in M}E^Pg_m \leq \sup\limits_{P \in M}E^Pf_m + \sup\limits_{P \in M}E^Pf_{m-1}<\infty, \ m=\overline{1, \infty}.$  The last  equality  and inequalities  give
\begin{eqnarray}\label{nick4}
f_m=f_0+\sum\limits_{i=1}^m f_{i-1}(\xi_i^0 -1) - \sum\limits_{i=1}^m g_{i}, \quad \ m=\overline{1, \infty}.
\end{eqnarray}
Let us consider  the random process  $\{M_m, {\cal F}_m\}_{m=0}^\infty,$ where $M_m=f_0+\sum\limits_{i=1}^m f_{i-1}(\xi_i^0 -1).$ Then  $E^P\{M_m| {\cal F}_{m-1}\}$ $=M_{m-1}, \ P\in M, \ m=\overline{1, \infty}.$
Theorem \ref{nick1} is proved.
\end{proof}

\subsection{Space of  finite set of elementary events.}
In this subsection we assume that  a space of elementary events  $\Omega$ is finite one, that is,  $N_0=|\Omega|< \infty,$  and we give a new proof of the optional decomposition for  super-martingales relative to the complete convex set of equivalent measures. This proof  does not use topological arguments as  in  \cite{WalterSchacher}.

 Let  ${\cal F}$ be   a certain algebra of subsets of the set   $\Omega$  and let  ${\cal F}_n \subset {\cal F}_{n+1} \subset {\cal F} $ be an increasing set of  algebras, where   ${\cal F}_0 =\{\emptyset, \Omega\}, $
  ${\cal F}_N = {\cal F}. $ Denote  $M$ a convex set of equivalent measures on a measurable space $\{\Omega, {\cal F}\}.$ Further, we assume that the set $A_0$ contains  an element $\xi_0\neq 1.$  It is evident that every algebra ${\cal F}_n$ is generated by sets $A_i^n, \   i=\overline{1, N_n}, A_i^n\cap A_j^n=\emptyset, \ i \neq j, \ N_n<\infty, \ \bigcup\limits_{i=1}^{N_n}A_i^n=\Omega, \ n=\overline{1,N}.$
Let  $m_n=E^P\{\xi_0|{\cal F}_n\}, \ P \in M,  \   n=\overline{1, N}.$ Then for $m_n$ the representation 
\begin{eqnarray}\label{1myk1}
m_n=\sum\limits_{i=1}^{N_n}m_i^n\chi_{A_i^n}(\omega), \quad n=\overline{1,N},
\end{eqnarray} 
is valid.
Consider the difference $d^n(\omega)=m_n -m_{n-1}.$ 
Then 
\begin{eqnarray}\label{1myk3}
d^n(\omega)=\sum\limits_{j=1}^{N_n}d_{j}^n\chi_{A_j^n}(\omega)=\sum\limits_{j\in I^-_n}d_{j}^n\chi_{A_j^n}(\omega)+\sum\limits_{j\in I^+_n}d_{j}^n\chi_{A_j^n}(\omega), 
\end{eqnarray}
\begin{eqnarray}\label{1myk4}
\sum\limits_{j\in I^-_n}\chi_{A_j^n}(\omega)+
\sum\limits_{j\in I^+_n}\chi_{A_j^n}(\omega)=1, 
\end{eqnarray}
where $d_j^n\leq 0,$ as $ j \in I^-_n,$ and $d_j^n> 0$ for $j \in I^+_n.$
From the equalities (\ref{1myk3}), (\ref{1myk4})   we obtain
\begin{eqnarray}\label{1myk5}
E^Pd^n(\omega)=\sum\limits_{j\in I^-_n}d_{j}^nP(A_j^n)+\sum\limits_{j\in I^+_n}d_{j}^nP(A_j^n)=0, \quad P  \in M,
\end{eqnarray}
\begin{eqnarray}\label{1myk6}
\sum\limits_{j\in I^-_n}P(A_j^n)+
\sum\limits_{j\in I^+_n}P(A_j^n)=1, \quad  \in M.
\end{eqnarray}
Denote  $M_n$ the contraction of the set of measures $M$ on the algebra ${\cal F}_n.$  Introduce into the set $M_n $  the metrics
\begin{eqnarray}\label{1myk7} 
\rho_n(P_1,P_2)=\max\limits_{B}\sum\limits_{s=1}^k|P_1(B_s^n) - P_2(B_s^n)|, \ P_1, P_2 \in M_n,  \ n=\overline{1,N},
\end{eqnarray}
where $B=\{B_1^n, \ldots, B_k^n\}$  is a partition of $\Omega$ on $k$ subsets,  that is,  $B_i^n \in {\cal F}_n, \  i=\overline{1,k}, \
 B_i^n \cap B_j^n=\emptyset, i\neq j, \bigcup\limits_{i=1}^k B_i^n =\Omega.$
The maximum  in the formula (\ref{1myk7}) 
is all over the partitions of the set  $\Omega,$ belonging to the $\sigma$-algebra ${\cal F}_n.$

\begin{defin}\label{1myk8} On a measurable space
$ \{\Omega, {\cal F}\},$  a convex  set of equivalent  measure $M$ we call complete if for every $1 \leq n \leq N$
the closure of the set of measures $M_n$ in the metrics (\ref{1myk7}) contains the measures
\begin{eqnarray}\label{1myk9} 
P_{ij}^n(A)= \left\{\begin{array}{l l} 0, & A\neq A_i^n, A_j^n,\\
                                                    \frac{d_j^n}{-d_i^n +d_j^n}, & A=A_i^n,\\
                                                      \frac{-d_i^n}{-d_i^n +d_j^n}, & A=A_j^n
                                                     \end{array}\right.       
\end{eqnarray}
for every  $i \in I^-_n$ and $ j \in I^+_n.$
\end{defin}

\begin{lemma}\label{1myk10} Let a convex family of equivalent measures $M$ be a complete one and the set $A_0$ contains an element $\xi_0\neq 1.$ Then for every non negative ${\cal F}_n$-measurable random value $\xi_n=\sum\limits_{i=1}^{N_n} C_i^n \chi_{A_i^n}$ there exists a real number $\alpha_n$ such that
\begin{eqnarray}\label{1myk11} 
\frac{ \sum\limits_{i=1}^{N_n} C_i^n \chi_{A_i^n} }{\sup\limits_{P \in M_n}\sum\limits_{i=1}^{N_n} C_i^n P(A_i^n)} \leq 1+\alpha_n (m_n - m_{n-1}), \quad 
n=\overline{1,N}.
\end{eqnarray}
\end{lemma}
 \begin{proof} 
On the set  $\bar M_n,$  the  functional $\varphi(P)=\sum\limits_{i=1}^{N_n} C_i^n P(A_i^n)$  is a continuous one, where  $ \bar M_n$ is the closure of the set $M_n$ in the metrics $\rho_n(P_1,P_2).$
From this it follows that  the equality
\begin{eqnarray}\label{1myk13}
\sup\limits_{P \in M_n}\sum\limits_{i=1}^{N_n} C_i^n P(A_i^n)=\sup\limits_{P \in \bar M_n}\sum\limits_{i=1}^{N_n} C_i^n P(A_i^n)
\end{eqnarray} 
is valid.
 Denote  $f_i^n=\frac{C_i^n}{\sup\limits_{P \in M_n}\sum\limits_{i=1}^{N_n} C_i^n P(A_i^n)}, \ i=\overline{1, N_n}.$ Then
\begin{eqnarray}\label{1myk14}
\sum\limits_{i=1}^{N_n} f_i^n P(A_i^n) \leq 1, \quad P \in \bar M_n.
\end{eqnarray} 
 For those $i \in I^-_n$ for which $d_i^n<0$ and those $j \in I^+_n$ for which $d_j^n>0$  the inequality (\ref{1myk14}) is as follows 
\begin{eqnarray}\label{1myk15}
\quad  f_i^n \frac{d_j^n}{-d_i^n +d_j^n} + \frac{-d_i^n}{-d_i^n+d_j^n}f_j^n\leq 1,\  d_i^n <0,  \  i \in I^-,  \ d_j ^n>0,  \  j \in I^+_n. 
\end{eqnarray} 
From  (\ref{1myk15})  we obtain the inequalities
\begin{eqnarray}\label{1myk16}
f_j^n \leq 1+\frac{1- f_i^n}{-d_i^n}d_j^n, \quad d_i^n <0,  \quad i \in I^-_n,  \quad d_j^n >0,  \quad  j \in I^+_n.
\end{eqnarray} 
Since the inequalities  (\ref{1myk16}) are valid for every $\frac{1- f_i^n}{-d_i^n},$ as $ d_i^n <0, $ and since the set of such  elements is finite,  then if to denote
\begin{eqnarray}\label{apm29}
 \alpha_n =\min_{\{i, \ d_i^n <0\}}\frac{1- f_i^n}{-d_i^n},
\end{eqnarray}
 then we have
\begin{eqnarray}\label{1myk17}
f_j^n \leq 1+\alpha_n d_j^n,   \quad  d_j^n >0,  \quad  j \in I^+_n.
\end{eqnarray} 
From the definition of  $\alpha_n$ we obtain the inequalities 
\begin{eqnarray}\label{apm30}
f_i^n \leq 1+\alpha_n d_i^n,  \quad   d_i^n <0, \quad  i \in I^-_n.
\end{eqnarray} 
Now if $d_i^n=0$ for some $ i \in I^-_n,$ then in this case $f_i^n \leq 1.$ All these inequalities give 
\begin{eqnarray}\label{1myk18}
f_i^n \leq 1+\alpha_n d_i^n,  \quad   i \in I^-_n\cup I^+_n.
\end{eqnarray} 
   Multiplying on $\chi_{A_i^n}$ the inequalities (\ref{1myk18}) and summing over all $ i \in I^-_n\cup I^+_n$ we obtain the needed inequality.
Lemma  \ref{1myk10}    is proved.
\end{proof}

\begin{thm}\label{1myk19} Suppose that the conditions of  Lemma \ref{1myk10}
are valid. Then  every non negative super-martingale  $\{f_m, {\cal F}_m\}_{m=0}^N$   relative to a convex set of equivalent measures $M,$  satisfying conditions
\begin{eqnarray}\label{apm31}
\frac{f_n}{f_{n-1}}\leq C_n<\infty, \quad n=\overline{1, N},
\end{eqnarray}
 is a local regular one, where $C_n, \ n=\overline{1, N}, $ are constants.
\end{thm}
\begin{proof} Consider the random value $\xi_n=\frac{f_n}{f_{n-1}}.$ Due to Lemma \ref{1myk10}
\begin{eqnarray}\label{apm32}
 \frac{\xi_n}{\sup\limits_{P \in M}E^P\xi_n} \leq 1+\alpha_n(m_n-m_{n-1})=\xi_n^0, \quad  n=\overline{1,N}.
\end{eqnarray}
It is evident that  $E^P\{\xi_n^0|{\cal F}_{n-1}\}=1, \ P \in M, \ n=\overline{1, N}.$ 
Since $\sup\limits_{P \in M} E^P\xi_n \leq 1,$ then 
\begin{eqnarray}\label{1myk20}
 \frac{f_n}{f_{n-1}} \leq \xi_n^0, \quad n=\overline{1, N}.
\end{eqnarray} 
Theorem \ref{nick1} and the inequalities (\ref{1myk20}) prove  Theorem \ref{1myk19}.
\end{proof}

\begin{thm}\label{myktinal4}
On a finite space of elementary events $\{\Omega, {\cal F}\}$ with a filtration  ${\cal F}_n$ on it,  every  super-martingale  $\{f_m, {\cal F}_m\}_{m=0}^N$ relative to the complete convex  set of equivalent  measures $M$ is a local regular one if the set $A_0$ contains $\xi_0\neq 1.$
\end{thm}
\begin{proof} It is evident that every super-martingale  $\{f_m, {\cal F}_m\}_{m=0}^N$ is bounded. Therefore, there exists a constant $C_0>0$ such that
$\frac{3C_0}{2}> f_m+C_0>\frac{C_0}{2}, \ \omega \in \Omega, \ m=\overline{0,N}.$ From this it follows that the super-martingale  $\{f_m+C_0, {\cal F}_m\}_{m=0}^N$ is a nonnegative one
and satisfies the conditions
\begin{eqnarray}\label{apm33}
 \frac{f_n+C_0}{f_{n-1}+C_0}\leq 3, \quad n=\overline{1, N}.
\end{eqnarray}
It implies that the conditions of Theorem \ref{1myk19} are satisfied.
Theorem \ref{myktinal4} is proved.
\end{proof}

\begin{thm}\label{myktinal1}
Let $M$ be a complete convex set of equivalent measure on a measurable space $\{\Omega, {\cal F}\}$ with a filtration ${\cal F}_m$ on it. Suppose that $\xi_0 \in A_0, \  \xi_0\neq 1,$ and $m_n=E^P\{\xi_0|{\cal F}_n\}$  is a martingale relative to the set of measures $M.$  Let $M_0^a$ be a set of all martingale measures absolutely continuous relative to any measure $P \in M.$ Then the inclusion $\bar M \subseteq M_0^a$ is valid, where $\bar M$ is a closure of the set of measures $M$ in metrics $\rho_N(P_1, P_2),$ defined in (\ref{1myk7}). 
\end{thm}
\begin{proof} Let  the sequence  $P_s \in M$ be a convergent one to the measure $P_0 \in \bar M,$ then for $D \in {\cal F}_{n-1}$
\begin{eqnarray}\label{myktinal2}
 \int\limits_{D}m_n d P_s= \int\limits_{D}m_{n-1} d P_s, \quad s=\overline{1, \infty}.
\end{eqnarray}
The functionals $ \int\limits_{D}m_n d P, \  \int\limits_{D}m_{n-1} d P$  on the set $\bar M$ for all 
$D \in {\cal F}_{n-1}$ are continuous ones   relative to the metrics $\rho_N(P_1, P_2),$  defined by the formula (\ref{1myk7}). Going  to the limit in the equality (\ref{myktinal2}), as $s \to \infty,$  we obtain
\begin{eqnarray}\label{myktinal3}
 \int\limits_{D}m_n d P_0= \int\limits_{D}m_{n-1} d P_0, \quad n=\overline{1,N}, \quad D \in {\cal F}_{n-1}. 
\end{eqnarray}
The last implies that $P_0 \in M_0^a.$  Theorem \ref{myktinal1} is proved.
\end{proof}

%%%%%%%%%%%%%%%%%%%%%%%%%%%%%%%%%%%%%%%%%%%%%
\subsection{Countable set of elementary events.}
In this subsection, we generalize the results of the previous subsection onto the countable space of elementary events.
Let  ${\cal F}$ be  a certain $\sigma$-algebra of subsets of  the countable set of elementary events  $\Omega$  and let  ${\cal F}_n \subset {\cal F}_{n+1} \subset {\cal F} $ be   a certain  increasing set  of $\sigma$-algebras, where   ${\cal F}_0 =\{\emptyset, \Omega\}.$
 Denote  $M$ a set of equivalent measures on the measurable space $\{\Omega, {\cal F}\}.$  Further, we assume that the set $A_0$ contains an element $\xi_0\neq 1.$  Suppose that the  $\sigma$-algebra ${\cal F}_n$ is generated by the sets $A_i^n, \   i=\overline{1, \infty}, \  A_i^n\cap A_j^n=\emptyset, \ i \neq j, \ \bigcup\limits_{i=1}^{\infty}A_i^n=\Omega, \ n=\overline{1,\infty}.$

Introduce into consideration the martingale  $m_n=E^P\{\xi_0|{\cal F}_n\}, \ P \in M, \  n=\overline{1, \infty}.$ Then for $m_n$ the representation 
\begin{eqnarray}\label{2myk1}
m_n=\sum\limits_{i=1}^{\infty}m_i^n\chi_{A_i^n}(\omega), \quad n=\overline{1,\infty},
\end{eqnarray} 
is valid.
Consider the difference $d^n(\omega)=m_n -m_{n-1}.$ 
Then 

\begin{eqnarray}\label{2myk3}
d^n(\omega)=\sum\limits_{j=1}^{\infty}d_{j}^n\chi_{A_j^n}(\omega) =\sum\limits_{j\in I^-}d_{j}^n\chi_{A_j^n}(\omega)+\sum\limits_{j\in I^+}d_{j}^n\chi_{A_j^n}(\omega), 
\end{eqnarray}
\begin{eqnarray}\label{2myk4}
\sum\limits_{j\in I^-}\chi_{A_j^n}(\omega)+
\sum\limits_{j\in I^+}\chi_{A_j^n}(\omega)=1, 
\end{eqnarray}
where  $d_j^n\leq 0,$ as $ j \in I^-_n,$ and  $d_j^n> 0,$  $j \in I^+_n.$
From the equalities (\ref{2myk3}), (\ref{2myk4}) we obtain
\begin{eqnarray}\label{2myk5}
E^Pd^n(\omega)=\sum\limits_{j\in I^-_n}d_{j}^nP(A_j^n)+\sum\limits_{j\in I^+_n}d_{j}^nP(A_j^n)=0, \quad P  \in M,
\end{eqnarray}
\begin{eqnarray}\label{2myk6}
\sum\limits_{j\in I^-_n}P(A_j^n)+
\sum\limits_{j\in I^+_n}P(A_j^n)=1, \quad P  \in M.
\end{eqnarray}
Denote  $M_n$ the contraction of the set of measures $M$ on the $\sigma$-algebra ${\cal F}_n.$  Introduce into the set $M_n $  the metrics
\begin{eqnarray}\label{2myk7} 
\rho_n(P_1,P_2)=\sup\limits_{B}\sum\limits_{s=1}^k|P_1(B_s^n) - P_2(B_s^n)|, \ P_1, P_2 \in M_n, \ n=\overline{1,\infty},
\end{eqnarray}
where $B=\{B_1^n, \ldots, B_k^n\}$  is a partition of $\Omega$ on $k$ subsets,  that is,  $B_i^n \in {\cal F}_n, \  i=\overline{1,k}, \
 B_i^n \cap B_j^n=\emptyset, i\neq j, \bigcup\limits_{i=1}^k B_i^n =\Omega.$
The supremum in the formula (\ref{2myk7} ) is  all over the partitions of the set  $\Omega,$ belonging to the $\sigma$-algebra ${\cal F}_n.$
\begin{defin}\label{2myk8} On a measurable space
$ \{\Omega, {\cal F}\}$ with a filtration ${\cal F}_n$ on it,  a convex  set of equivalent measure $M$  we call complete one if for every $1 \leq n < \infty$
the closure of the set of measures $M_n$ in the metrics (\ref{2myk7}) contains the measures
\begin{eqnarray}\label{2myk9} 
P_{ij}^n(A)= \left\{\begin{array}{l l} 0, & A\neq A_i^n, A_j^n,\\
                                                    \frac{d_j^n}{-d_i^n +d_j^n}, & A=A_i^n,\\
                                                      \frac{-d_i^n}{-d_i^n +d_j^n}, & A=A_j^n
                                                     \end{array}\right.       
\end{eqnarray}
for every  $i \in I^-_n$ and $ j \in I^+_n.$
\end{defin}
\begin{lemma}\label{2myk10} Let a family of measures $M$ be complete and the set $A_0$ contains an element $\xi_0 \neq 1.$ Then for every non negative bounded ${\cal F}_n$-measurable random value $\xi_n=\sum\limits_{i=1}^{\infty} C_i^n \chi_{A_i^n}$ there exists a real number $\alpha_n$ such that
\begin{eqnarray}\label{2myk11} 
\frac{ \sum\limits_{i=1}^{\infty} C_i^n \chi_{A_i^n} }{\sup\limits_{P \in M_n}\sum\limits_{i=1}^{\infty} C_i^n P(A_i^n)} \leq 1+\alpha_n (m_n - m_{n-1}), \quad 
n=\overline{1,\infty}.
\end{eqnarray}
\end{lemma}
\begin{proof}  On the set  $\bar M_n,$ the functional $\varphi(P)=\sum\limits_{i=1}^{\infty} C_i^n P(A_i^n)$  is a  continuous  one relative to the metrics $\rho_n(P_1,P_2),$ where  $ \bar M_n$ is the closure of the set $M_n$ in this metrics.
From this it follows that  the equality
\begin{eqnarray}\label{2myk13}
\sup\limits_{P \in M_n}\sum\limits_{i=1}^{\infty} C_i^n P(A_i^n)=\sup\limits_{P \in \bar M_n}\sum\limits_{i=1}^{\infty} C_i^n P(A_i^n)
\end{eqnarray} 
is valid.
Denote  $f_i^n=\frac{C_i^n}{\sup\limits_{P \in M_n}\sum\limits_{i=1}^{\infty} C_i^n P(A_i^n)}, \ i=\overline{1, \infty}.$
Then
\begin{eqnarray}
\sum\limits_{i=1}^{\infty} f_i^n P(A_i^n) \leq 1, \quad P \in \bar M_n.
\end{eqnarray}
The last inequalities can be written  in the form
\begin{eqnarray}\label{2myk14}
\sum\limits_{i \in I^-} f_i^n P(A_i^n) +\sum\limits_{i \in I^+}f_i^n P(A_i^n) \leq 1, \quad P \in \bar M_n.
\end{eqnarray}
 For those $i \in I^-_n$ for which $d_i^n<0$ and those $j \in I^+_n$ for which $d_j^n>0$  the inequality (\ref{2myk14}) is as follows 
\begin{eqnarray}\label{2myk15}
 f_i^n \frac{d_j^n}{-d_i^n +d_j^n} + \frac{-d_i^n}{-d_i^n+d_j^n} f_j^n\leq 1, \
 d_i^n <0, \ d_j ^n>0, \  i \in I^-_n, \ j \in I^+_n.
\end{eqnarray} 

From  (\ref{2myk15})  we obtain the inequalities
\begin{eqnarray}\label{2myk16}
f_j^n \leq 1+\frac{1- f_i^n}{-d_i^n}d_j^n, \quad d_i^n <0, \quad  d_j^n >0, 
\quad  i \in I^-_n, \quad j \in I^+_n.
\end{eqnarray} 
Two cases are possible: a) for all  $ i \in I^-_n,$ $ f_i^n \leq 1; $ b) there exists $ i \in I^-_n$ such that $  f_i^n > 1.$
First, let us consider the case a).

Since the inequalities  (\ref{2myk16}) are valid for every $\frac{1- f_i^n}{-d_i^n},$ as $ d_i^n <0, $ and $ f_i^n \leq 1, i \in I^-_n ,$  then if to denote
\begin{eqnarray}
 \alpha_n =\inf_{\{i, \ d_i^n <0\}}\frac{1- f_i^n}{-d_i^n},
\end{eqnarray}
 we have $ 0 \leq \alpha_n < \infty$  and
\begin{eqnarray}\label{2myk17}
f_j^n \leq 1+\alpha_n d_j^n,  \quad   d_j^n >0,  \quad j \in I^+_n.
\end{eqnarray} 
From the definition of  $\alpha_n$ we obtain the inequalities 
\begin{eqnarray}
f_i^n \leq 1+\alpha_n d_i^n,  \quad   d_i^n <0, \quad  i \in I^-_n.
\end{eqnarray} 
Now, if $d_i^n=0$ for some $ i \in I^-_n, $ then in this case $f_i^n \leq 1.$ All these inequalities give 
\begin{eqnarray}\label{2myk18}
f_i^n \leq 1+\alpha_n d_i^n,  \quad   i \in I^-_n\cup I^+_n.
\end{eqnarray} 
Consider the case b). From the inequality (\ref{2myk16}), we obtain
\begin{eqnarray}\label{2myk19}
f_j^n \leq 1-\frac{1- f_i^n}{d_i^n}d_j^n, \quad d_i^n <0, \quad  d_j^n >0, \quad  i \in I^-_n, \quad j \in I^+_n.
\end{eqnarray} 
 The last inequalities give
\begin{eqnarray}\label{2myk20}
\frac{1- f_i^n}{d_i^n} \leq \min_{\{j, \ d_j^n >0\}} \frac{1}{d_j^n}< \infty, \quad  d_i^n <0, \quad  i \in I^-_n.
\end{eqnarray}
Let us define $\alpha_n =\sup\limits_{\{ i, \ d_i^n<0 \} }\frac{1- f_i^n}{d_i^n}< \infty.$
Then from (\ref{2myk19}) we obtain
\begin{eqnarray}\label{2myk21}
f_j^n \leq 1- \alpha_n d_j^n, \quad  \  d_j^n >0, \quad j \in I^+_n.
\end{eqnarray}
From the definition of  $\alpha_n, $ we have
\begin{eqnarray}\label{2myk22}
f_i^n \leq 1- \alpha_n d_i^n, \quad  \  d_i^n <0, \quad  i \in I^-_n.
\end{eqnarray}
The inequalities (\ref{2myk21}), (\ref{2myk22}) give
\begin{eqnarray}\label{2myk23}
f_j^n \leq 1- \alpha_n d_j^n,  \quad  j \in I^-_n\cup I^+_n.
\end{eqnarray}

   Multiplying on $\chi_{A_i^n}$ the inequalities (\ref{2myk18}) and  the inequalities (\ref{2myk23})  on $\chi_{A_j^n}$ and summing over all $i, j \in I^-_n\cup I^+_n$ we obtain the needed inequality.
The Lemma  \ref{2myk10}  is proved.
\end{proof}

\begin{thm}\label{2myk24} Suppose that the conditions of  Lemma \ref{2myk10}
are valid. Then  every non negative super-martingale  $\{f_m, {\cal F}_m\}_{m=0}^\infty $  relative to a convex set of equivalent measures $M,$ satisfying the conditions 
\begin{eqnarray}\label{2myk25}
 \frac{f_m}{f_{m-1}} \leq C_m< \infty, \quad m=\overline{1,\infty},
\end{eqnarray}
 is a local regular one, where $C_m$ are constants.
\end{thm}
\begin{proof} 
From the conditions (\ref{2myk25}) it follows that  $\sup\limits_{P \in M}E^Pf_m < \infty.$ 
Consider the random value $\xi_n=\frac{f_n}{f_{n-1}}.$ Due to Lemma \ref{2myk10}
\begin{eqnarray}
 \frac{\xi_n}{\sup\limits_{P \in M}E^P\xi_n} \leq 1+\alpha_n(m_n-m_{n-1})=\xi_n^0.
\end{eqnarray}
It is evident that  $E^P\{\xi_n^0|{\cal F}_{n-1}\}=1, \ P \in M, \ n=\overline{1, \infty}.$ 
Since $\sup\limits_{P \in M}E^P\xi_n \leq 1,$ then 
\begin{eqnarray}\label{2myk26}
 \frac{f_n}{f_{n-1}} \leq \xi_n^0, \quad n=\overline{1, \infty}.
\end{eqnarray} 
Theorem \ref{nick1} and the inequalities (\ref{2myk26}) prove  Theorem \ref{2myk24}.
\end{proof}

\subsection{An arbitrary space of elementary events.}
In this subsection, we consider an arbitrary space of elementary events and prove the optional decomposition for non negative super-martingales.

Let  ${\cal F}$ be a certain $\sigma$-algebra of subsets of  the set of elementary events  $\Omega$  and let  ${\cal F}_n \subset {\cal F}_{n+1} \subset {\cal F} $ be an increasing set of the $\sigma$-algebras, where   ${\cal F}_0 =\{\emptyset, \Omega\}.$
 Denote  $M$ a set of equivalent measures on a measurable space $\{\Omega, {\cal F}\}.$ We assume that the $\sigma$-algebras  ${\cal F}_n, \  n=\overline{1, \infty},$ and ${\cal F} $  are complete relative to any measure $P \in M.$   
Further, we suppose that the set $A_0$ contains an element $\xi_0\neq 1.$ 
Let  $m_n=E^P\{\xi_0|{\cal F}_n\}, \ P \in M, \ n=\overline{1, \infty}.$ 

Consider the difference $d^n(\omega)=m_n -m_{n-1}.$ We assume that  every $\omega \in \Omega$ belongs to the $\sigma$-algebra ${\cal F}_n, \ n=\overline{1, \infty},$ and $P(\{\omega\})=0, \ \omega \in\Omega, \ P \in M.$

 For the random value  $d^n(\omega)$ there exists not more then a countable set of  the real number $d_s^n $ such that $P(A_s^n)>0,$  where  $A_s^n=\{\omega \in \Omega, d^n(\omega)=d_s^n\}.$  It is evident that  $A_i^n\cap A_j^n=\emptyset, \ i \neq j.$  Suppose that 
$P(\Omega \setminus \bigcup\limits_{i=1}^{\infty}A_i^n)>0.$ Introduce for every $n$ two subsets  $I^-_n=\{\omega \in \Omega, \ d^n(\omega)\leq 0\},$
$I^+_n=\{\omega \in \Omega, \ d^n(\omega)> 0\}$ of the set $\{\omega\in \Omega, \  |d_n(\omega)|<\infty\}.$ 

Denote  $M_n$ the contraction of the set of measures $M$ on the $\sigma$-algebra ${\cal F}_n.$  Introduce into the set $M_n $  the metrics
\begin{eqnarray}\label{3myk7} 
\rho_n(P_1,P_2)=\sup\limits_{B}\sum\limits_{s=1}^k|P_1(B_s^n) - P_2(B_s^n)|, \ P_1, P_2 \in M_n, \ n=\overline{1,\infty},
\end{eqnarray}
where $B=\{B_1^n, \ldots, B_k^n\}$  is a partition of $\Omega$ on $k$ subsets,  that is,  $B_i^n \in {\cal F}_n, \  i=\overline{1,k}, \
 B_i^n \cap B_j^n=\emptyset, i\neq j, \bigcup\limits_{i=1}^k B_i^n =\Omega.$
The supremum in the formula (\ref{3myk7} ) is  all over the partitions of the set  $\Omega,$ belonging to the $\sigma$-algebra ${\cal F}_n.$

\begin{defin}\label{3myk8}   On a measurable space
$ \{\Omega, {\cal F}\}$ with filtration ${\cal F}_n$ on it, a convex  set of equivalent measure $M$ we call complete if for every $1 \leq n < \infty$   the closure   in metrics (\ref{3myk7}) of the  set of measures $M_n$ 
 contains the measures
\begin{eqnarray}\label{3myk9} 
P_{\omega_1, \omega_2}^n(A)= \left\{\begin{array}{l l} 0, &\omega_1, \  \omega_2  \in \Omega\setminus A,\\
                                                    \frac{d^n(\omega_2)}{-d^n(\omega_1) +d^n(\omega_2)}, & \omega_1 \in  A, \ A\cap \{\omega_2\}=\emptyset,\\
                                                      \frac{-d^n(\omega_1)}{-d^n(\omega_1) +d^n(\omega_2)}, &  \omega_2 \in\Omega \setminus A, \ (\Omega \setminus A)\cap \{\omega_1\}=\emptyset
                                                     \end{array}\right.       
\end{eqnarray}
for   $\omega_1 \in I^-_n$ and $ \omega_2 \in I^+_n.$
\end{defin}

\begin{lemma}\label{3myk10} Let a convex  family of equivalent measures $M$ be a complete one and the set $A_0$ contains an element $\xi_0 \neq 1.$ Then for every non negative bounded ${\cal F}_n$-measurable random value $\xi_n$ there exists a real number $\alpha_n$ such that
\begin{eqnarray}\label{3myk11}
\frac{\xi_n }{\sup\limits_{P \in M}E^P\xi_n} \leq 1+\alpha_n (m_n - m_{n-1}), \quad 
n=\overline{1,\infty}.
\end{eqnarray}
\end{lemma}
\begin{proof} On the set  $\bar M_n,$ the functional $\varphi(P)=\int\limits_{\Omega}\xi_n d P$  is a  continuous  one relative to the metrics $\rho_n(P_1,P_2),$ where  $ \bar M_n$ is the closure of the set $M_n$ in this metrics.
From this it follows that  the equality
\begin{eqnarray}\label{3myk13}
\sup\limits_{P \in M_n}\int\limits_{\Omega}\xi_n d P=\sup\limits_{P \in \bar M_n}\int\limits_{\Omega}\xi_n d P
\end{eqnarray} 
is valid.
Denote  $f^n(\omega)=\frac{\xi_n(\omega)}{\sup\limits_{P \in M_n}E^P\xi_n(\omega)}.$
Then
\begin{eqnarray}
E^Pf^n(\omega) \leq 1, \quad P \in \bar M_n.
\end{eqnarray}
The last inequalities can be written  in the form
\begin{eqnarray}\label{3myk14}
 \int\limits_{I^-}f_n(\omega) d P+ \int\limits_{I^+}f_n(\omega)  d P \leq 1, \quad P \in \bar M_n.
\end{eqnarray}
The inequality (\ref{3myk14}) for the measures   (\ref{3myk9})    is as follows 
$$ f^n(\omega_1) \frac{d^n(\omega_2)}{-d^n(\omega_1) +d^n(\omega_2)} +$$
\begin{eqnarray}\label{3myk15}
 \frac{-d^n(\omega_1)}{-d^n(\omega_1)+
d^n(\omega_2)} f^n(\omega_2)\leq 1, \quad \omega_1 \in I^-_n, \quad \omega_2 \in I^+_n.
\end{eqnarray} 
From  (\ref{3myk15})  we obtain the inequalities
\begin{eqnarray}\label{3myk16}
f^n (\omega_2)\leq 1+\frac{1- f^n(\omega_1)}{-d^n(\omega_1)}d^n(\omega_2), \end{eqnarray} 
\begin{eqnarray}
 d^n(\omega_1) <0, \quad  d^n(\omega_2) >0, 
\quad  \omega_1 \in I^-_n, \quad \omega_2 \in I^+_n.
\end{eqnarray}
Two cases are possible: a) for all  $ \omega_1 \in I^-_n,$ $ f^n(\omega_1) \leq 1; $ b) there exists $ \omega_1 \in I^-_n$ such that $  f^n (\omega_1)> 1.$
First, let us consider the case a).

Since the inequalities  (\ref{3myk16}) are valid for every $\frac{1- f^n(\omega_1)}{-d^n(\omega_1)},$ as $ d^n(\omega_1) <0, $ and $ f^n(\omega_1) \leq 1, \omega_1 \in I^-_n ,$  then if to denote
\begin{eqnarray}
\alpha_n =\inf_{\{\omega_1, \ d^n(\omega_1) <0\}}\frac{1- f^n(\omega_1)}{-d^n(\omega_1)},
\end{eqnarray}
 we have $ 0 \leq \alpha_n < \infty$  and
\begin{eqnarray}\label{3myk17}
f^n (\omega_2)\leq 1+\alpha_n d^n(\omega_2),  \quad   d^n (\omega_2)>0,  \quad \omega_2 \in I^+_n.
\end{eqnarray} 
From the definition of  $\alpha_n$ we obtain the inequalities 
\begin{eqnarray}
f^n(\omega_1)  \leq 1+\alpha_n d^n(\omega_1),  \quad   d^n(\omega_1) <0, \quad  \omega_1 \in I^-_n.
\end{eqnarray} 
Now, if $d^n(\omega_1)=0$ for some $ \omega_1\in I^-_n, $ then in this case $f^n(\omega_1) \leq 1.$ All these inequalities give 
\begin{eqnarray}\label{3myk18}
f^n(\omega) \leq 1+\alpha_n d^n(\omega),  \quad \omega \in  I^-_n\cup  I^+_n.
\end{eqnarray} 
Consider the case b). From the inequality (\ref{3myk16}), we obtain
\begin{eqnarray}\label{3myk19}
f^n(\omega_2) \leq 1-\frac{1- f^n(\omega_1)}{d^n(\omega_1)}d^n(\omega_2),
\end{eqnarray}
\begin{eqnarray}
  d^n(\omega_1) <0, \quad  d^n(\omega_2) >0, \quad  \omega_1 \in I^-_n, \quad \omega_2 \in I^+_n.
\end{eqnarray}
 The last inequalities give
\begin{eqnarray}\label{3myk20}
\frac{1- f^n(\omega_1)}{d^n(\omega_1)} \leq \inf_{\{\omega_2, \ d^n(\omega_2) >0\}} \frac{1}{d^n(\omega_2)}< \infty, \quad  d^n(\omega_1) <0, \quad  \omega_1 \in I^-_n.
\end{eqnarray}
Let us define $\alpha_n =\sup\limits_{\{ \omega_1, \ d^n(\omega_1)<0 \} }\frac{1- f^n(\omega_1)}{d^n(\omega_1)}< \infty.$
Then from (\ref{3myk19}) we obtain
\begin{eqnarray}\label{3myk21}
f^n(\omega_2)   \leq 1- \alpha_n d^n(\omega_2), \quad  \  d^n(\omega_2) >0, \quad \omega_2 \in I^+_n.
\end{eqnarray}
From the definition of  $\alpha_n $ we have
\begin{eqnarray}\label{3myk22}
f^n(\omega_1) \leq 1- \alpha_n d^n(\omega_1), \quad  \  d^n(\omega_1) <0, \quad  \omega_1 \in I^-_n.
\end{eqnarray}
The inequalities (\ref{3myk21}), (\ref{3myk22}) give
\begin{eqnarray}\label{3myk23}
f^n(\omega) \leq 1- \alpha_n d^n(\omega),  \quad  \omega \in  I^-_n\cup  I^+_n.
\end{eqnarray}
Since the set $ I^-_n\cup  I^+_n$ has probability one,
Lemma  \ref{3myk10}  is proved.
\end{proof}
\begin{thm}\label{3myk24} Suppose  a  convex set of equivalent measures   $M$ is a complete one and the conditions of  Lemma \ref{3myk10}
are valid. Then  every non negative super-martingale  $\{f_m, {\cal F}_m\}_{m=0}^\infty $ 
 relative to a convex set of equivalent measures $M,$
satisfying conditions 
\begin{eqnarray}\label{3myk25}
  \frac{f_m}{f_{m-1}} \leq C_m< \infty, \quad m=\overline{1,\infty},
\end{eqnarray}
 is a local regular one, where $C_m,  m=\overline{1,\infty},$ are constants.
\end{thm}
\begin{proof} From the inequalities (\ref{3myk25}) it follows that $\sup\limits_{P \in M}E^Pf_m<\infty, \ m=\overline{1,\infty}.$
Consider the random value $\xi_n=\frac{f_n}{f_{n-1}}.$ Due to Lemma \ref{3myk10}
\begin{eqnarray}
 \frac{\xi_n}{\sup\limits_{P \in M}E^P\xi_n} \leq 1+\alpha_n(m_n-m_{n-1})=\xi_n^0.
\end{eqnarray}
It is evident that  $E^P\{\xi_n^0|{\cal F}_{n-1}\}=1, \ P \in M, \ n=\overline{1, \infty}.$ 
Since $\sup\limits_{P \in M}E^P\xi_n \leq 1,$ then 
\begin{eqnarray}\label{3myk26}
 \frac{f_n}{f_{n-1}} \leq \xi_n^0, \quad n=\overline{1, \infty}.
\end{eqnarray} 
Theorem \ref{nick1} and the inequalities (\ref{3myk26}) prove  Theorem \ref{3myk24}.
\end{proof}
\begin{ce}\label{Myktin}
If a super-martingale  $\{f_m, {\cal F}_m\}_{m=0}^\infty$ relative to a complete convex set of equivalent measures $M$   satisfy conditions $0 \leq f_m \leq D_m,\ m=\overline{1,\infty},$  where $D_m< \infty$ are constant, then it is local regular.
\end{ce}
\begin{proof} The super-martingale $\{f_m+\varepsilon, {\cal F}_m\}_{m=0}^\infty,$  $\varepsilon>0,$ is a nonnegative one and satisfies the conditions 
\begin{eqnarray}
\frac{f_m+\varepsilon}{f_{m-1}+\varepsilon}\leq \frac{D_m+\varepsilon}{\varepsilon}=C_m< \infty, \quad  m=\overline{1,\infty}.
\end{eqnarray}
From Theorem \ref{2myk24} it follows the validity of the local regularity for the super-martingale $\{f_m+\varepsilon, {\cal F}_m\}_{m=0}^\infty,$
 therefore, for the super-martingale $\{f_m, {\cal F}_m\}_{m=0}^\infty$ the  local regularity is also true.
\end{proof}

\section{Local regularity of  majorized  super-martingales.}
In this section, we give the elementary proof that a majorized super-martingale relative to the complete set of equivalent measures  is local regular one.

\begin{thm}\label{Tinmyk1}
On a measurable space $\{\Omega, {\cal F}\}$
 with a filtration ${\cal F}_m$ on it, let the set  $M$ be a complete convex set of equivalent measures on ${\cal F}$ and the set $A_0$ contains an element $\xi_0\neq 1.$ Then every bounded super-martingale $\{f_m, {\cal F}_m\}_{m=0}^\infty $ relative to  the complete convex  set of equivalent measures $M$  is a local regular one.
\end{thm}
\begin{proof}  From Theorem \ref{Tinmyk1} conditions, there exists a constant  $0<C<\infty$ such that  $|f_m| \leq C, \ m=\overline{1, \infty}.$  Consider  the super-martingale $\{f_m+C, {\cal F}_m\}_{m=0}^\infty .$ Then $0\leq f_m+C\leq 2C.$  Due to Consequence \ref{Myktin},
for the super-martingale $\{f_m+C, {\cal F}_m\}_{m=0}^\infty $ the local regularity is true. So, the same statement is valid  for the super-martingale $\{f_m, {\cal F}_m\}_{m=0}^\infty. $ Theorem \ref{Tinmyk1} is proved.
\end{proof}
 The next Theorem is  analogously proved as Theorem \ref{Tinmyk1} .
\begin{thm}\label{Tinmyk2}
On a measurable space $\{\Omega, {\cal F}\}$
 with filtration ${\cal F}_m$ on it, let the set  $M$ be a complete convex set of equivalent measures on ${\cal F}$ and the set $A_0$ contains an element $\xi_0\neq 1.$ Then a  super-martingale $\{f_m, {\cal F}_m\}_{m=0}^\infty $ relative to  the complete convex  set of equivalent measures $M$ satisfying the conditions 
\begin{eqnarray}
|f_m|\leq C_1 \xi_0, \quad f_m+C_1 \xi_0\leq C_2, \quad m=\overline{1,\infty}, \quad \xi_0 \in A_0,  
\end{eqnarray}
 for certain constants  $0<C_1, C_2<\infty,$  is  a local regular one.
\end{thm}

\section{ Application to Mathematical Finance.}

Due to Corollary \ref{mars16}, we can give the following definition of the fair price of contingent claim $f_N$ relative to a convex set of equivalent measures $M.$
\begin{defin}\label{maras1}
Let $f_N, \ N< \infty, $ be a ${\cal F}_N$-measurable  integrable  random value  relative to a convex set of equivalent measures  $M$   such that for some  $0 \leq \alpha_0< \infty$ and $\xi_0 \in A_0$
 \begin{eqnarray}\label{maras2}
P(f_N - \alpha_0 E^P\{\xi_0|{\cal F}_N\} \leq 0)=1.
\end{eqnarray}
Denote $ G_{\alpha_0}=\{\alpha \in [0,  \alpha_0],  \ \exists \xi_{\alpha} \in A_0, \ P(f_N - \alpha E^P\{\xi_{\alpha}|{\cal F}_N\} \leq 0)=1\}.$
We call
\begin{eqnarray}\label{maras4}
f_0=\inf\limits_{\alpha \in G_{\alpha_0}}\alpha
\end{eqnarray}
  the fair price of the  contingent claim  $f_N$ relative to a convex set of equivalent measures $M,$ if there exists $\zeta_0 \in A_0$  and a sequences $\alpha_n \in [0,\alpha_0],$  $\xi_{\alpha_n} \in A_0,$  satisfying the conditions: $\alpha_n \to f_0,$ 
$\xi_{\alpha_n} \to \zeta_0$ by probability, as $n \to \infty,$ and such that
\begin{eqnarray}\label{maras3}
P(f_N - \alpha_n E^P\{\xi_{\alpha_n}|{\cal F}_N\} \leq 0)=1, \quad n=\overline{1, \infty}.
\end{eqnarray}
\end{defin}
\begin{thm}\label{mars17}
Let the set $A_0$ be uniformly integrable one relative to every measure $P \in M.$ Suppose that  for  a nonnegative ${\cal F}_N$-measurable  integrable  contingent claim $f_N, \  N< \infty,$ relative to every measure $P \in M$   there exist $\alpha_0< \infty$ and $\xi_0 \in A_0$ such that
\begin{eqnarray}\label{mars18}
P(f_N - \alpha_0 E^P\{\xi_0|{\cal F}_N\} \leq 0)=1,
\end{eqnarray}
then the fair price $f_0$ of contingent claim $f_N$ exists.
For $f_0$ the inequality 
\begin{eqnarray}\label{mars20}
 \sup\limits_{P \in M}E^Pf_N \leq f_0
\end{eqnarray}
is valid. If $f_N \geq 0$ and  a super-martingale $\{f_m=\mathrm{ess}\sup\limits_{P \in M}E^P\{f_N|{\cal F}_m\},{\cal F}_m\}_{m=0}^\infty$ is a local regular one, then $f_0=\sup\limits_{P \in M}E^Pf_N.$ 
\end{thm}
\begin{proof} If $ f_0=\alpha_0,$ then Theorem \ref{mars17} is proved. Suppose that 
$ f_0<\alpha_0.$ Then there exists  a sequence $\alpha_n \to f_0,$ and $\xi_{\alpha_n} \in A_0, \  n \to \infty,$ such that
\begin{eqnarray}\label{mars21}
 P(f_N - \alpha_n E^P\{\xi_{\alpha_n}|{\cal F}_N\} \leq 0)=1, \quad P \in M.
\end{eqnarray}
 Due to the uniform integrability $ A_0$ we obtain
\begin{eqnarray}\label{mars22}
1=\lim\limits_{n \to \infty}\int\limits_{\Omega}\xi_{\alpha_{n}}dP=\int\limits_{\Omega} \zeta_0 dP, \quad P \in M.
\end{eqnarray}
 Using again the uniform integrability of $ A_0$ and going to the limit in   (\ref{mars21}) we obtain
\begin{eqnarray}\label{mars23}
 P(f_N - f_0 E^P\{ \zeta_0|{\cal F}_N\} \leq 0)=1, \quad P \in M.
\end{eqnarray}
From the inequality $ f_N - f_0 E^P\{ \zeta_0|{\cal F}_N\} \leq 0$ it follows
the  inequality (\ref{mars20}).
If $f_N \geq 0$ and  $\{f_m=\mathrm{ess}\sup\limits_{P \in M}E^P\{f_N|{\cal F}_m\}, {\cal F}_m\}_{m=0}^N $ is a local regular super-martingale, then
\begin{eqnarray}\label{bars1}
f_m=M_m- g_m, \quad m=\overline{0,N}, \quad g_0=0,
\end{eqnarray}
where a martingale $\{M_m, {\cal F}_m\}_{m=0}^N$ is a nonnegative one and $E^P M_m=\sup\limits_{P \in M}E^P f_N.$ Introduce  into consideration a random value $\xi_0=\frac{M_N}{\hat f_0},$ where $ \hat f_0= \sup\limits_{P \in M}E^Pf_N. $ Then $\xi_0$ belongs to the set $A_0$ and 
\begin{eqnarray}\label{bars2}
P(f_N - \hat f_0  E^P\{ \xi_0|{\cal F}_N\} \leq 0)=1.	
\end{eqnarray}
From this it follows that
 $f_0=\sup\limits_{P \in M}E^Pf_N.$ 

Let us prove that $f_0$ is a fair price for  certain evolutions of risk and non risk assets.
Suppose that the evolution of risk asset is given by the law $S_m=f_0M^P\{\zeta_0|{\cal F}_m\}, \ m=\overline{0,N}, $ and the evolution of non risk asset is given by the formula  $B_m=1,  \ m=\overline{0,N}.$

As proved above,  for $f_0=\inf\limits_{\alpha \in G_{\alpha_0}}\alpha $
there exists $ \zeta_0 \in A_0$  such that the inequality 
\begin{eqnarray}
f_N - f_0 E^P\{\zeta_0|{\cal F}_N\} \leq 0
\end{eqnarray}
is valid. Let us put
\begin{eqnarray} 
f_m^0=f_0 E^P\{\zeta_0|{\cal F}_m\},  \quad P \in M,
\end{eqnarray}
\begin{eqnarray}
\bar f_m=
\left\{\begin{array}{l l} 0, & m<N, \\
f_N -f_0 E^P\{\zeta_0|{\cal F}_m\}, & m = N.  
 \end{array} \right. 
\end{eqnarray}
It is evident that $\bar f_{m-1} -\bar f_m \geq 0, \  m=\overline{0, N}.$
Therefore, the super-martingale
\begin{eqnarray}
f_m^0+ \bar f_m=
\left\{\begin{array}{l l} f_0 E^P\{\zeta_0|{\cal F}_m\}, & m<N, \\
f_N , & m= N,  
\end{array} \right. 
\end{eqnarray}
is a local regular one.
It is evident that
\begin{eqnarray}
f_m^0+ \bar f_m=M_m - g_m, \quad  m=\overline{0, N},
\end{eqnarray}
where 
\begin{eqnarray}
 M_m=f_0 E^P\{\zeta_0|{\cal F}_m\},  \quad  m=\overline{0, N},
\end{eqnarray}
\begin{eqnarray}
 g_m = 0, \quad  \ m=\overline{0, N-1}, 
\end{eqnarray}
\begin{eqnarray}
 g_ N =f_0 E^P\{\zeta_0|{\cal F}_N\} - f_N.
\end{eqnarray}
  For the martingale  $\{M_m, {\cal F}_m\}_{m=0}^N$ the representation
\begin{eqnarray}
 M_m=f_0+\sum\limits_{i=1}^m H_i \Delta  S_i, \quad m=\overline{0, N},
\end{eqnarray}
 is valid, where $H_i=1, \  i=\overline{1, N}.$  
Let us consider the trading strategy  $\pi=\{\bar H_m^0, \bar H_m\}_{m=0}^N,$ where
\begin{eqnarray}
\bar H_0^0=f_0, \quad \bar H_m^0=M_m - H_m S_m, \quad m=\overline{1,N},
\end{eqnarray}
\begin{eqnarray}
 \quad  \bar H_0=0, \quad \bar H_m=H_m, \quad m=\overline{1,N}.
\end{eqnarray}
It is evident  that  $\bar H_m^0,  \bar H_m$ are  ${\cal F}_{m-1}$ measurable and the trading strategy  $\pi$ satisfy  self-financed condition
\begin{eqnarray}
 \Delta \bar H_m^0 +\Delta  \bar H_m S_{m-1}=0.
\end{eqnarray}
Moreover, the capital corresponding to the  self-financed trading strategy  $\pi$ is given by the formula
\begin{eqnarray}
X_m^{\pi}=\bar H_m^0+  \bar H_m  S_{m} = M_m.
\end{eqnarray}
Herefrom,  $X_0^{\pi}=f_0.$
Further,
\begin{eqnarray}
 X_N^{\pi}=f_N+g_N \geq f_N.
\end{eqnarray}
The last proves  Theorem \ref{mars17}.
\end{proof}

 From (\ref{mars23}) and Corollary \ref{mars16} the Theorem \ref{mars26}  follows.

\begin{thm}\label{mars26} 
Suppose that the set $A_0$ contains  only $ 1 \leq  k < \infty $ linear independent elements $\xi_1, \ldots \xi_k.$ If there exist $\xi_0 \in T$ and $\alpha_0\geq 0$ such that
\begin{eqnarray}\label{00mars27}
P(f_N - \alpha_0E^P\{ \xi_0 | {\cal F}_N \}\leq 0)=1, \quad P \in M,
\end{eqnarray}
where
\begin{eqnarray}\label{mars28}
T=\{\xi \geq 0, \  \xi=\sum\limits_{i=1}^k\alpha_i\xi_i, \ \alpha_i \geq 0, \ i=\overline{1,k}, \ \sum\limits_{i=1}^k\alpha_i=1\},
\end{eqnarray}
then the fair price $f_0$ of the contingent claim $f_N \geq 0$ exists, where $f_N$ is ${\cal F}_N$ measurable and integrable relative to every measure $P \in M,$ $N< \infty. $
\end{thm}
\begin{proof} The proof is evident, as the set $T$ is  a uniformly integrable one relative to every measure from $M.$
\end{proof}

\begin{cor}\label{rian1} On a measurable space $\{\Omega, {\cal F}\}$ with filtration ${\cal F}_m$ on it,  let  \\ $\{f_m, {\cal F}_m\}_{m=0}^N$ be a non negative  local regular super-martingale   relative to a convex set of equivalent measures $M.$   If the set $A_0$ is uniformly integrable relative to every measure $P \in M,$ then the fair price of contingent claim $f_N$ exists.
\end{cor}
\begin{proof} From the local regularity of super-martingale $\{f_m, {\cal F}_m\}_{m=0}^N$  we have $f_m=M_m - g_m, \ m=\overline{0,N}.$ Therefore, $P(f_N - \alpha_0 \xi_0 \leq 0)=1,$ where $\alpha_0=E^PM_N, \  P \in M, \xi_0=\frac{M_N}{E^PM_N}. $ From the last it follows that the conditions of Theorem \ref{mars17} are satisfied. Corollary \ref{rian1} is proved.
\end{proof}

On a probability space $\{\Omega, {\cal F}, P\},$ let us   consider  an evolution of  one risk asset    given by the law $ \{ S_m \}_{m=0}^N,$  where $S_m$ is a random  value taking values in $R_+^1.$  Suppose that
 ${\cal F}_m$ is a filtration on  $\{\Omega, {\cal F}, P\}$ and  $S_m$ is ${\cal F}_m$-measurable random value.
 We assume that the non risk asset evolve by the law $B_m^0=1,  \ m=\overline{1,N}. $
Denote $M^e(S)$ the set of all martingale measures being equivalent to the measure $P.$ 
 We assume that the set  $M^e(S)$ of such  martingale measures is not empty and the effective market is non complete, see, for example, \cite{Schacher1}, \cite{DMW90}, \cite{K81},  \cite{HK79}.
So, we have that 
\begin{eqnarray}\label{ma29}
E^Q\{S_m|{\cal F}_{m-1}\}=S_{m-1}, \quad m=\overline{1, N}, \quad Q \in M^e(S).
\end{eqnarray}

The next Theorem justify the Definition \ref{maras1}.

\begin{thm}\label{hon1} Let a contingent claim $f_N $  be  a 
${\cal F}_N $-measurable   integrable random value with respect to every measure from $M^e(S)$ and the  conditions of the Theorem  \ref{mars26} are satisfied with $\xi_i=\frac{S_i}{S_0}, \ i=\overline{0,N}.$ Then  there exists self-financed trading strategy $\pi$  the capital  evolution $\{X_m^{\pi}\}_{m=0}^N$ of which  is  a martingale relative to every measure from $M^e(S)$  satisfying conditions  $ X_0^{\pi}=f_0, \  X_N^{\pi} \geq f_N,$
where  $f_0$ is a fair price of contingent claim $f_N. $ 
\end{thm}
\begin{proof} Due to Theorems \ref{mars17}, \ref{mars26}, for $f_0=\inf\limits_{\alpha \in G_{\alpha_0}}\alpha $
%\begin{eqnarray}\label{ghon1}
%\end{eqnarray}
there exists $ \zeta_0 \in A_0$  such that the inequality 
\begin{eqnarray}\label{2ghon2}
f_N - f_0 E^P\{\zeta_0|{\cal F}_N\} \leq 0
\end{eqnarray}
is valid. Let us put 
\begin{eqnarray}
f_m^0=f_0 E^P\{\zeta_0|{\cal F}_m\},  \quad P \in M^e(S),
\end{eqnarray}
\begin{eqnarray}
\bar f_m=
\left\{\begin{array}{l l} 0, & m<N, \\
f_N -f_0 E^P\{\zeta_0|{\cal F}_m\}, & m = N.  
 \end{array} \right. 
\end{eqnarray}
It is evident that $\bar f_{m-1} -\bar f_m \geq 0, \  m=\overline{0, N}.$
Therefore, the super-martingale
\begin{eqnarray}
f_m^0+ \bar f_m=
\left\{\begin{array}{l l} f_0 E^P\{\zeta_0|{\cal F}_m\}, & m<N, \\
f_N , & m= N  
\end{array} \right. 
\end{eqnarray}
is a local regular one.
It is evident that
\begin{eqnarray}
f_m^0+ \bar f_m=M_m - g_m, \quad  m=\overline{0, N},
\end{eqnarray}
where 
\begin{eqnarray}
 M_m=f_0 E^P\{\zeta_0|{\cal F}_m\},  \quad  m=\overline{0, N},
\end{eqnarray}
\begin{eqnarray}
g_m = 0, \quad  \ m=\overline{0, N-1},
\end{eqnarray}
\begin{eqnarray} 
 g_ N =f_0 E^P\{\zeta_0|{\cal F}_N\} - f_N.
\end{eqnarray}
Due to Theorem \ref{t8},  for the martingale  $\{M_m\}_{m=0}^N$ the representation
\begin{eqnarray}
 M_m=f_0+\sum\limits_{i=1}^m H_i \Delta  S_i, \quad m=\overline{0, N},
\end{eqnarray}
 is valid. 
Let us consider the trading strategy  $\pi=\{\bar H_m^0, \bar H_m\}_{m=0}^N,$ where
\begin{eqnarray}
\bar H_0^0=f_0, \quad \bar H_m^0=M_m -  H_m   S_m, \quad m=\overline{1,N},
\end{eqnarray}
\begin{eqnarray}
 \quad  \bar H_0=0, \quad \bar H_m=H_m, \quad m=\overline{1,N}.
\end{eqnarray}
It is evident  that  $\bar H_m^0,  \bar H_m$ are  ${\cal F}_{m-1}$-measurable ones and the trading strategy  $\pi$ satisfy the self-financed condition
\begin{eqnarray}
 \Delta \bar H_m^0 +\Delta  \bar H_m  S_{m-1}=0.
\end{eqnarray}
Moreover, a capital corresponding to the  self-financed trading strategy  $\pi$ is given by the formula
\begin{eqnarray}
X_m^{\pi}=\bar H_m^0+  \bar H_m  S_{m} = M_m.
\end{eqnarray}
Herefrom,  $X_0^{\pi}=f_0.$
Further,
\begin{eqnarray}
X_N^{\pi}=f_N+g_N.
\end{eqnarray}
Therefore $X_N^{\pi} \geq f_N.$ Theorem \ref{hon1}  is proved.
\end{proof}
In the next Theorem we assume that the evolutions of risk and non risk assets generate  incomplete  market  \cite{Schacher1},   \cite{DMW90}, \cite{K81}, \cite{HK79}, \cite{HP81}, that is, the set of martingale measures contains more that one element.

\begin{thm}\label{mars30}
Let  an evolution $ \{ S_m \}_{m=0}^N$ of the risk asset  satisfy the conditions $P(D_m^1 \leq S_m \leq D_m^2)=1, $ where the constants $D_m^i$ satisfy the inequalities $ D_{m-1}^1 \geq  D_m^1>0, \  D_{m-1}^2 \leq D_{m}^2< \infty, \ m=\overline{1,N},$ and let the non risk asset evolution be deterministic one given  by the law $\{B_m\}_{m=0}^N, \ B_m=1, \ m=\overline{0, N}.$ The fair price of Standard European Call Option  with the payment function $f_N=(S_N - K)^+$
is given by the formula
\begin{eqnarray}\label{mars31}
f_0=\left\{\begin{array}{l l} S_0(1-\frac{K}{D_N^2}), & K \leq D_N^2, \\
 0, & K >D_N^2. \end{array}\right.
\end{eqnarray}

The fair price of Standard European Put Option  with the payment function $f_N=(K-S_N)^+$
is given by the formula
\begin{eqnarray}\label{mmars31}
f_0=\left\{\begin{array}{l l} K- D_N^1, & K \geq D_N^1, \\
 0, & K < D_N^1. \end{array}\right.
\end{eqnarray}
\end{thm}
\begin{proof} In  Theorem \ref{mars30} conditions, the set of equations $E^P\zeta=1, \ \zeta \geq 0, $ has the solutions $\zeta_i=\frac{S_i}{S_0}, \ i=\overline{0,N}.$ It is evident that  $\alpha_0=S_0$ and $\zeta_N=\frac{S_N}{S_0},$ since
\begin{eqnarray}
 \frac{(S_N - K)^+}{B_N} - \alpha_0 \frac{S_N}{S_0} \leq 0, \quad \omega \in \Omega.
\end{eqnarray}
Let us prove the needed formula. Consider the inequality
\begin{eqnarray}\label{mars33}
  (S_N - K) - \alpha \sum\limits_{i=0}^N\gamma_i\frac{S_i}{S_0} \leq 0, \quad  \gamma \in V_0,
\end{eqnarray}
where  $V_0=\{\gamma=\{\gamma_i\}_{i=0}^N, \ \gamma_i \geq 0, \ \sum\limits_{i=0}^N\gamma_i=1\}.$
Or,
\begin{eqnarray}\label{mars34}
 S_N\left ( 1 - \frac{ \alpha\gamma_N}{S_0}\right) - K - \alpha \sum\limits_{i=0}^{N-1}\gamma_i\frac{S_i}{S_0} \leq 0.
\end{eqnarray}
Suppose that  $\alpha$ satisfies the inequality
\begin{eqnarray}\label{mars35}
 1 -  \frac{ \alpha}{S_0}>0.
\end{eqnarray}
If $\alpha $ satisfies additionally the equality
\begin{eqnarray}\label{mars36}
 D_N^2\left( 1 - \frac{ \alpha\gamma_N}{S_0}\right) - K - \alpha \sum\limits_{i=0}^{N-1}\gamma_i\frac{D_i^1}{S_0}= 0,
\end{eqnarray}
  then for all $\omega \in \Omega$ 
(\ref{mars34}) is valid.  From (\ref{mars36}) we obtain for $\alpha$
\begin{eqnarray}\label{mars37}
\alpha=\frac{S_0(D_N^2 - K)}{(D_N^2\gamma_N+ \sum\limits_{i=0}^{N-1}\gamma_i D_i^1)}.
\end{eqnarray}
If $D_N^2 -K >0,$ then
\begin{eqnarray}\label{mars38}
\inf\limits_ {\gamma \in V_0}\frac{S_0(D_N^2 - K)}{(D_N^2\gamma_N+ \sum\limits_{i=0}^{N-1}\gamma_i D_i^1)}=\frac{S_0(D_N^2 - K)}{D_N^2},
\end{eqnarray}
since $ D_N^2 \geq D_i^1.$
From here we obtain
\begin{eqnarray}\label{mars39}
f_0=S_0\left(1 - \frac{ K}{D_N^2}\right).
\end{eqnarray}
It is evident that $\alpha = f_0$ satisfies the inequality (\ref{mars35}).

If  $D_N^2 - K \leq 0,$ then $S_N - K \leq 0$  and from (\ref{mars33}) we can put $\alpha=0.$ Then, the formula (\ref{mars34}) is valid for all $\omega \in \Omega.$

Let us prove the formula (\ref{mmars31}) for Standard European Put Option.
If $S_N \leq K$ it is evident that  $\alpha_0=K, $ and $\zeta_0=1,$ since
\begin{eqnarray}
 (K -S_N) - \alpha_0  \leq 0, \quad \omega \in \Omega.
\end{eqnarray} 
Let us prove the needed formula. 
 Consider the  inequality
\begin{eqnarray}\label{mmars33}
  (K- S_N)^+ - \alpha \sum\limits_{i=0}^N\gamma_i\frac{S_i}{S_0}  \leq 0,  \quad  \gamma \in V_0.
\end{eqnarray}
Or, for $S_N \leq K$
\begin{eqnarray}\label{mmars34}
 - S_N\left( 1 + \frac{ \alpha \gamma_N}{S_0}\right) + K -\alpha \sum\limits_{i=0}^{N-1}\gamma_i\frac{S_i}{S_0}\leq 0.
\end{eqnarray}
If  $\alpha$ is a solution of the equality
\begin{eqnarray}\label{mmars36}
 - D_N^1\left( 1 + \frac{ \alpha \gamma_N}{S_0}\right) + K - \alpha \sum\limits_{i=0}^{N-1}\gamma_i\frac{D_i^1}{S_0}  = 0,
\end{eqnarray}
 then for all $\omega \in \Omega$ 
(\ref{mmars34}) is valid.  From (\ref{mmars36}) we obtain for $\alpha$
\begin{eqnarray}\label{mmars37}
\alpha=\frac{S_0( K - D_N^1)}{\sum\limits_{i=0}^{N}\gamma_i D_i^1 }.
\end{eqnarray}
Therefore,
\begin{eqnarray}\label{mmars38}
\inf\limits_ {\gamma \in V_0}\frac{S_0(K- D_N^1)}{\sum\limits_{i=0}^{N}\gamma_i D_i^1}= K - D_N^1,
\end{eqnarray}
since $ D_i^1\leq S_0, \ i=\overline{1,N}, \ D_0^1=S_0.$
From here we obtain
\begin{eqnarray}\label{mmars39}
f_0= K - D_N^1.
\end{eqnarray}
If  $D_N^1 - K > 0,$ then $S_N - K > 0$  and from (\ref{mmars33}) we can put $\alpha=0.$ Then, (\ref{mmars34}) is valid for all $\omega \in \Omega.$
The Theorem  \ref{mars30} is proved.
\end{proof}

\section{Some auxiliary results.}
On a measurable space $\{\Omega, {\cal F}\}$ with filtration ${\cal F}_n$ on it,  let us consider  a convex set of equivalent measures $M.$
Suppose that  $\xi_1, \ldots, \xi_d$ is a set of random values belonging to the set $A_0.$ 
Introduce $d$ martingales relative to a set of measures $M$ $\{S_n^i, {\cal F}_n\}_{n=0}^{\infty}, \  i=\overline{1,d},$    where  $S_n^i=E^P\{\xi_i|{\cal F}_n\}, \  i=\overline{1,d}, \ P \in M. $    Denote by    $M^{e}(S)$   a set of all martingale measures equivalent to  a measure  $P \in M,$   that is,  $Q \in M^{e}(S)$ if
\begin{eqnarray}
 E^Q\{S_n|{\cal F}_{n-1}\}=S_{n-1},\ E^Q|S_n|< \infty,  \ Q \in M^{e}(S), \ n=\overline{1, \infty}. 
\end{eqnarray} 
It is evident that  $ M \subseteq M^{e}(S) $  and  $M^{e}(S)$ is a convex set.
Denote   $ P_0$  a certain fixed measure from $M^{e}(S)$  and let  $L^0(R^d)$   be a set of finite valued random values on  a probability space $\{ \Omega, {\cal F}, P_0\},$ taking values in $R^d.$

 Let    $H^0$ be  a set of finite valued predictable processes $H=\{H_n\}_{n=1}^N, $ where $H_n=\{H_n^i\}_{i=1}^d$ takes values in $R^d$  and  $ H_n$ is ${\cal F}_{n-1}$-measurable random vector.
Introduce into consideration a set of random values  
\begin{eqnarray}\label{s5}
K_N^1=\{ \xi \in L^0(R^1), \  \xi= \sum\limits_{k=1}^N   \langle H_k,\Delta S_k \rangle, , \ H  \in H^0\}, \  N< \infty,
\end{eqnarray}
\begin{eqnarray}
\Delta S_k = S_k - S_{k-1}, \quad \langle H_k, \Delta S_k \rangle = \sum\limits_{s=1}^dH_k^s(S_k^s-S_{k-1}^s).
\end{eqnarray} 

\begin{lemma}\label{q7}
The set of random values  $K_N^1$  is  a closed subset in the set of finite valued random values  $L^0(R^1)$  relative to the convergence by measure $P \in M.$
\end{lemma}
 The proof of the Lemma \ref{q7} see, for example, \cite{DMW90}.

Introduce into consideration a subset
\begin{eqnarray}
 V^0 =\{H \in H^0,\ ||H_n||< \infty,\ n=\overline{1, N}\} 
\end{eqnarray} 
of the set    $H^0,$ where $||H_n||=\sup\limits_{\omega \in \Omega}\sum\limits_{i=1}^d|H_n^i|.$ 
Let $K_N$ be a subset of the set $K_N^1$
\begin{eqnarray}
K_N=\{\xi \in L^0(R^1), \  \xi=\sum\limits_{k=1}^N   \langle H_k,\Delta S_k \rangle, \ H \in V^0 \}.
\end{eqnarray}
Denote also a set
\begin{eqnarray} 
C =\{k-f,\ k \in  K_N, \ f \in L_+^{\infty}(\Omega, {\cal F}, P_0)\},
\end{eqnarray}
where  $L_+^{\infty}(\Omega, {\cal F}, P_0\}$  is a set of bounded nonnegative random values. Let   $\bar C$ be   the closure of  $C$ in $L^{1}(\Omega, {\cal F}, P_0)$ metrics.
\begin{lemma}\label{l8} If  $\zeta \in \bar C$  and such that  $E^{P_0}\zeta=0,$ then for   $\zeta$ the representation 
\begin{eqnarray} 
\zeta=\sum\limits_{k=1}^N \langle H_k,\Delta S_k \rangle 
\end{eqnarray}
is valid for a certain finite valued predictable process
  $H=\{H_n\}_{n=1}^N.$
\end{lemma}
\begin{proof} If $\zeta \in K_N,$  then Lemma \ref{l8} is proved. Suppose that   $\zeta \in \bar C,$  then there exists a sequence  $k_n - f_n, \ k_n \in K_N,\ f_n \in  L_+^{\infty}(\Omega, {\cal F}, P_0)$ such that 
$||k_n - f_n - \zeta ||_{P_0} \to 0,\ n \to \infty,$ where $||g||_{P_0}=E^{P_0}|g|.$ Since 
$|E^{P_0}(k_n - f_n - \zeta)| \leq ||k_n - f_n - \zeta ||_{P_0},$ we have $E^{P_0}f_n  \leq ||k_n - f_n - \zeta ||_{P_0}.$
From here we obtain    $||k_n - \zeta ||_{P_0} \leq 2 ||k_n - f_n - \zeta ||_{P_0}.$ Therefore, $k_n \to \zeta$
by measure $P_0.$  On the basis of  Lemma  \ref{q7}, a set 
\begin{eqnarray}
 K_N^1=\{\xi \in L^0(R^1), \  \xi=\sum\limits_{k=1}^N \langle H_k,\Delta S_k \rangle, \ H \in H^0 \}, 
\end{eqnarray}
\begin{eqnarray}
  \langle H_k,\Delta S_k \rangle=\sum\limits_{i=1}^dH_k^i(S_k^i - S_{k-1}^i)
\end{eqnarray}
is  a closed subset of $L^0(R^1)$  relative to the convergence by measure  $P_0.$ 
From this fact, we obtain the proof of Lemma  \ref{l8}, since there exists the finite valued predictable process  $H \in H^0$ such that for   $\zeta $ the representation 
\begin{eqnarray}
\zeta=\sum\limits_{k=1}^N\langle H_k,\Delta S_k \rangle 
\end{eqnarray}
is valid.
\end{proof}
\begin{thm}\label{t7} Let  $E^Q|\zeta| < \infty, \ Q \in  M^{e}(S).$ If for every    $\ Q \in  M^{e}(S), \ E^Q\zeta = 0,$ then there exists  finite valued predictable process  $H$ such that for  $ \zeta$ the representation  
\begin{eqnarray}\label{m1}
 \zeta=\sum\limits_{k=1}^N \langle H_k,\Delta S_k \rangle
\end{eqnarray}
is valid.
\end{thm}
\begin{proof}  If  $\zeta  \in \bar C,$ then  (\ref{m1}) follows from Lemma  \ref{l8}. So, let  $\zeta$  does not belong to $\bar C.$  As in Lemma  \ref{l8}, $\bar C$  is a closure of  $C$ in  $L^{1}(\Omega, {\cal F}, P_0)$  metrics for the  fixed measure  $P_0.$ The set  $\bar C$ is a closed convex set in $L^{1}(\Omega, {\cal F}, P_0).$   Consider the  other convex closed set that consists from one element   $\zeta. $ Due to  Han -- Banach  Theorem, there exists a linear continuous functional $l_1,$ which belongs to  $L^{\infty}(\Omega, {\cal F}, P_0),$  and real numbers  $\alpha > \beta$ such that 
\begin{eqnarray}\label{m2}
l_1(\xi)=\int\limits_{\Omega}\xi(\omega) q(\omega)dP_0, \quad q(\omega) \in  L^{\infty}(\Omega, {\cal F}, P_0),
\end{eqnarray}
and the inequalities  $l_1(\zeta) > \alpha,$ $l_1(\xi) \leq \beta, \ \xi \in  \bar C,$  are valid.  Since   $\bar C$  is a convex cone we can put $ \beta=0.$  From the condition  $l_1(\xi) \leq 0, \ \xi \in  \bar C$ we have $l_1(\xi)=0, \  \xi \in K_N^1 \cap L^{1}(\Omega, {\cal F}, P_0).$ From  (\ref{m2}) and the inclusions $ \bar C \supset C \supset - L^{\infty}(\Omega, {\cal F}, P_0)$ we have  $ q(\omega) \geq 0.$  Introduce  a measure
\begin{eqnarray}
Q^*(A)= \int\limits_{A} q(\omega)dP_0\left[\int\limits_{\Omega} q(\omega)dP_0\right]^{-1}.
\end{eqnarray}
Then, we have 
\begin{eqnarray}\label{m3}
  \int\limits_{\Omega} \xi(\omega)dQ^*=0,\quad \xi \in K_N^1 \cap L^{1}(\Omega, {\cal F}, P_0).
\end{eqnarray}
Let us choose  $ \xi=\chi_{A}(\omega)(S_i^j -S_{i-1}^j ), \ A \in {\cal F}_{i-1},$ where  $\chi_{A}(\omega)$ is an indicator of a set   $A.$  We obtain 
\begin{eqnarray}
\int\limits_{A}(S_i^j -S_{i-1}^j )dQ^*=0, \quad  A \in {\cal F}_{i-1}.
\end{eqnarray}
So, $Q^*$ is a martingale measure that belongs to the set  $M^a(S),$ which is a set of absolutely continuous martingale measures.  Let us choose   $Q \in  M^e(S)$ and consider a measure  $Q_1=(1- \gamma)Q+\gamma Q^* , \ 0 < \gamma < 1.$ A measure   $ Q_1 \in  M^e(S) $ and, moreover, $E^{Q_1}\zeta=\gamma E^{Q^*}\zeta>0.$  We come to the contradiction  with the conditions of Theorem  \ref{t7}, since for  $Q \in  M^e(S), \ E^Q\zeta=0.$ So, $\zeta \in \bar C,$  and in accordance with  Lemma  \ref{l8}, for  $\zeta$ the  declared  representation in Theorem \ref{t7}  is valid.
\end{proof}

\begin{thm}\label{t8}
 For every martingale  $\{M_n,{\cal F}_n\}_{n=0}^{\infty}$ relative to the set of measures  $ M^e(S),$ there exists a predictable random process
$H$ such that for   $M_n,\ n=\overline{0,\infty},$ the representation 
\begin{eqnarray}\label{m4}
M_n=M_0+\sum\limits_{i=1}^n\langle H_i,\Delta S_i \rangle, \quad  n=\overline{1,\infty},
\end{eqnarray}
is valid.
\end{thm}
\begin{proof}  For fixed natural  $N \geq 1,$ let us consider the random value $M_N - M_0=\zeta.$  Since 
\begin{eqnarray}
E^Q|\zeta|<\infty, \quad E^Q\zeta=0, \quad Q \in  M^e(S),
\end{eqnarray} 
 then 
$\zeta$ satisfies the conditions of Theorem \ref{t7}  and, therefore,  belongs to  $\bar C,$ so, there exists a sequence   
$k_n=\sum\limits_{i=1}^N \langle H_i^n, \Delta S_i \rangle \in K_N $ such that 
\begin{eqnarray}
\int\limits_{\Omega}|k_n - \zeta|dP_0 \to 0, \quad n \to \infty.
\end{eqnarray}
From here, we obtain 
\begin{eqnarray}
\int\limits_{\Omega}|E^{P_0}\{(k_n - \zeta)|{\cal F}_{m}\}|dP_0 \leq \int\limits_{\Omega}|k_n - \zeta|dP_0 \to 0,\ n \to \infty. 
\end{eqnarray}
But  $E^{P_0}\{k_n |{\cal F}_{m}\}=\sum\limits_{i=1}^m \langle H_i^n, \Delta S_i \rangle.$
Hence,  we obtain that   both  $\sum\limits_{i=1}^m\langle H_i^n, \Delta S_i\rangle $ and  $\sum\limits_{i=1}^N \langle H_i^n, \Delta S_i \rangle $
converges by measure  $P_0$ to  $E^{P_0}\{\zeta|{\cal F}_{m}\}$ and  $\zeta,$ correspondingly.
There exists a subsequence  $n_k$ such that $H^{n_k}$  converges everywhere to predictable  process   $H$.
From here, we have  $\zeta = \sum\limits_{i=1}^N\langle H_i, \Delta S_i\rangle $ and  $E^{P_0}\{\zeta|{\cal F}_{m}\}=\sum\limits_{i=1}^m\langle H_i, \Delta S_i\rangle.$  It proves that for all   $m< N$
\begin{eqnarray}
 M_m=M_0+\sum\limits_{i=1}^m\langle H_i, \Delta S_i \rangle.
\end{eqnarray}
Theorem  \ref{t8}  is proved.
\end{proof}

\section{Conclusions.}

In the paper, we generalize Doob  decomposition
  for super-martingales relative to one measure  onto the case of super-martingales relative to a convex set of equivalent measures.  For super-martingales relative to one measure for  continuous time Doob's result was generalized in papers 
\cite{Meyer1} \cite{Meyer2}.
Section 2 contains the definition of local regular super-martingales. Theorem  \ref{reww1} gives the necessary and sufficient conditions of the local regularity of  super-martingale. In spite of its simplicity,  the Theorem \ref{reww1} appeared very useful for the description of the local regular super-martingales.

For this purpose we investigate the structure of super-martingales of special types
relative to the convex set of equivalent measures, generated by a certain finite set of equivalent measures. 
The main result of the section 3    is   Lemma \ref{q5}, which allowed   proving Lemma \ref{1q5}, giving the sufficient conditions of the existence of a martingale  with respect to a convex set of equivalent measures generated by finite set of equivalent measures.

Theorem \ref{mars12} describes all local regular non negative super-martingales  of the special type (\ref{marsss13})  relative  to the convex set of equivalent measures, generated by the finite set of equivalent measures.

In the Theorem \ref{fmars5}, we give the sufficient conditions  of the existence  of the local regular martingale relative to an arbitrary set of equivalent measures and arbitrary filtration.  After that, we present in Theorem \ref{mmars1} the important  construction of the local regular super-martingales which  we sum up in Corollary \ref{fdr1}. Theorem 
\ref{9mmars9} proves that every majorized super-martingale  belongs to the  described class (\ref{mmars88}) of the local regular super-martingales. 

Theorem  \ref{nick1} gives a variant of the necessary and sufficient conditions of local regularity of non negative  super-martingale relative to a convex set of equivalent measures.    Definition \ref{1myk8} determines a class  of the complete set  of equivalent measures.   Lemma \ref{1myk10} guarantees   a bound  (\ref{1myk11})  for all non negative random values  allowing us  to prove   Theorem \ref{1myk19}, stating that for every  super-martingale the optional decomposition is valid.
 We extend the results obtained from the finite space of elementary events  onto the case as a space of elementary events is a countable one.  At last,  the subsection 5.3  contains the generalization of the result obtained in subsection 5.2 onto the case of arbitrary space of elementary events. In section 6,  we prove Theorems \ref{Tinmyk1} and \ref{Tinmyk2}, stating that for  every majorized super-martingale the optional decomposition is valid.

 Corollary \ref{mars16}   contains the important construction of the local regular super-martingales playing the important role in the definition of the fair price of contingent claim relative to a convex set of equivalent measures.
The Definition \ref{maras1} is a fundamental one for the evaluation of risks in  incomplete markets. 
Theorem \ref{mars17} gives the sufficient conditions of the existence of the fair price  of contingent claim relative to a convex set of equivalent measures. It also gives the sufficient conditions,  when the defined fair price coincides with the classical value. In  Theorem \ref{mars26} the simple conditions of the existence  of the fair price  of contingent claim  are given. In  Theorem \ref{hon1}  we  prove  the existence  of the self-financed trading  strategy confirming the Definition \ref{maras1} of the fair price  as  the parity between the long and short positions in contracts. 
As an application of the results obtained we prove Theorem \ref{mars30}, where the formulas for the Standard European Call and Put Options in an incomplete market we present. 
Section 8 contains auxiliary results needed for previous sections.

%%%%%%%%%%%%%%%%%%  Bibliography %%%%%%%%%%%%%%%%%%%%%
\def\cprime{$'$} \def\cprime{$'$}
{\color{Brown}
}
\end{document}